\documentclass[prx,twocolumn,english,superscriptaddress,floatfix,nofootinbib,aps,10pt]{revtex4-2} 
\usepackage{amsmath, amssymb, amscd, amsthm, amsfonts, amscd, dsfont, mathtools, wasysym, amsthm}
\usepackage{graphicx}%
\usepackage{dcolumn}%
\usepackage{bm}
\usepackage{standalone}
\usepackage{qcircuit}
\usepackage[caption=false]{subfig}
\usepackage[utf8]{inputenc}
\usepackage{microtype}
\usepackage[plain]{algorithm}
\usepackage{algpseudocode}
\newcommand{\algrule}[1][1pt]{\par\vskip.05\baselineskip\hrule height #1\par\vskip.15\baselineskip}
\usepackage{enumitem} 

\usepackage{silence}
\usepackage[dvipsnames]{xcolor}
\usepackage{hyperref}
\hypersetup{
    colorlinks = true,
    linkcolor = Blue,
    citecolor = MidnightBlue,
    urlcolor = cyan}

\newtheorem{theorem}{Theorem}
\newtheorem{new_th}{Theorem}

\usepackage{svg}

\DeclarePairedDelimiter\ket{|}{\rangle} 
\DeclarePairedDelimiter\bra{\langle}{|}
\newcommand{\braket}[2]{\langle #1 | #2 \rangle}
\newcommand{\ketbra}[2]{| #1 \rangle \langle #2 | }

\newcommand\tinyvarhexagon{\vcenter{\hbox{\scalebox{0.75}{$\varhexagon$}}}}

\begin{document}

\preprint{APS/123-QED}

\title{Dynamical Magic Transitions in Monitored Clifford\texorpdfstring{$\bm{+T}$}{} Circuits}%

\author{Mircea Bejan}
 \affiliation{T.C.M. Group, Cavendish Laboratory, University of Cambridge, J.J. Thomson Avenue, Cambridge, CB3 0HE, UK\looseness=-1}%
\author{Campbell McLauchlan}%
\affiliation{DAMTP, University of Cambridge, Wilberforce Road, Cambridge, CB3 0WA, UK}%
\author{Benjamin B\'eri}
\affiliation{T.C.M. Group, Cavendish Laboratory, University of Cambridge, J.J. Thomson Avenue, Cambridge, CB3 0HE, UK\looseness=-1}%
\affiliation{DAMTP, University of Cambridge, Wilberforce Road, Cambridge, CB3 0WA, UK}%

\begin{abstract}
The classical simulation of highly-entangling quantum dynamics is conjectured to be generically hard.
Thus, recently discovered measurement-induced transitions between highly-entangling and low-entanglement dynamics are phase transitions in classical simulability.
Here, we study simulability transitions beyond entanglement: noting that some highly-entangling dynamics (e.g., integrable systems or Clifford circuits) are easy to classically simulate, thus requiring
``magic"---a subtle form of quantum resource---to achieve computational hardness, we ask how the dynamics of magic competes with measurements.  
We study the resulting ``dynamical magic transitions" focusing on random monitored Clifford circuits doped by $T$ gates (injecting magic). 
We identify dynamical ``stabilizer-purification"---the collapse of a superposition of stabilizer states by measurements---as the mechanism driving this 
transition.
We find cases where transitions in magic and entanglement coincide, but also others with a magic and simulability transition in a highly (volume-law) entangled phase.
In establishing our results, we use Pauli-based computation, a scheme distilling the quantum essence of the dynamics to a magic state register subject to mutually commuting measurements. 
We link stabilizer-purification to ``magic fragmentation" wherein these measurements separate into disjoint, $\mathcal{O}(1)$-weight blocks, and relate this 
to the spread of magic in the original circuit becoming arrested. 
\end{abstract}
\maketitle

\section{Introduction}
The efficient simulation of generic quantum  systems is conjectured to require a quantum  computer~\cite{feynman1982simulating}. 
However, the boundary between what can  and cannot be efficiently simulated on a classical computer is a subtle issue~\cite{vidal2003efficient,vidal2004efficient,markov2008sim,schuch2008entropy,dalzell2022anticoncentration,napp2022shallow,wahl2022sim, Terhal2004ComplexitySampling,Feng2022SycamoreSampling,Boixo2018QuantumSupremacy,GoogleQuantumSupremacy2019}. 
The exponential dimension of the Hilbert space might naively suggest that an efficient simulation algorithm would be impossible in all but the most trivial of cases. 
Many recent experimental and theoretical efforts have confirmed the ability of random quantum circuits to generate output distributions that are exponentially complex to replicate classically~\cite{GoogleQuantumSupremacy2019,Boixo2018QuantumSupremacy,Quantum_supremacy2021,QRS23,Quantum_supremacy_boundary2023}.
Remarkably however, there exist examples of quantum dynamics that permit efficient classical simulation.
For example, it is possible to use matrix product states (MPSs) for efficiently simulating states with low entanglement~\cite{vidal2003efficient, vidal2004efficient}, Gaussian fermionic states for free-fermion dynamics~\cite{divincenzo_fermion_circuits2002,bravyi2004gaussian}, and the stabilizer formalism for Clifford dynamics~\cite{gottesman1998heisenberg, aaronson2004improved}.  
Delineating the boundary between quantum systems that do or do not permit efficient classical simulation can provide a greater understanding of the transition between quantum and classical dynamics and also expose the regimes in which future quantum computers could display an advantage. 
Entanglement is a resource for quantum advantage.
The existence of sharp transitions in the amount of entanglement generated by a quantum circuit~\cite{aharonov2000quantum} suggests the existence of a similar transition in classical simulation complexity. 
One mechanism for such entanglement transitions is via mid-circuit ``monitoring" measurements, mimicking the coupling of the system to an environment~\cite{li2018zeno, li2019mipt, choi2020qecc, jian2020mipt, ippoliti2021mom, fisher2023rqc, skinner2019mipt}.
In the highly-entangled phase, ``volume-law" scaling of the entanglement entropy (EE) is generated by the unitary gates in the circuit.
In the low-entanglement phase, randomly introduced monitor measurements suppress entanglement, resulting in an ``area-law" scaling.

The computational complexity of classically simulating this dynamics using MPSs is directly linked to EE~\cite{skinner2019mipt, bao2020tpt,Cirac2022TransitionsContinuousDynamics,Green2021TransitionsOpenSystems}. 
In the area-law regime, MPSs allow one to keep track of the system's state via polynomial-time (in the system size) classical computation, as opposed to exponential-time computations in the volume-law regime.
However, Clifford dynamics, which can be highly-entangling, can still be classically simulated efficiently.
This suggests there may be simulability transitions stemming from a quantum resource other than entanglement.
An important such resource is ``magic"~\cite{magic0, bravyi2005universal, magic2, Bravyi_PBC2016, magic1, magic3, magic4, magic5, magic6, magic7, Leone2021quantumchaosis}. This, broadly, quantifies how far a state is from the orbit of the Clifford group, if starting from a computational basis state.
Despite its importance, little is known about any magic-based simulability phase transition in monitored random quantum circuits.
Evidence exists for transitions in magic in certain random quantum circuits~\cite{magic8pt2, magic8pt1, magic9, singleT_univEE, SO_eCPX, magic_Pauli_rank},
and  magic has been studied in many-body states~\cite{CFT_magic,Ellison2021symmetryprotected,Magic_TFIM,Magic_MPS,Magic_MPS2,Magic_frust,Magic_thermalize,Magic_quench,Magic_chaos,Magic_CFT2,Magic_RK, Magic_multifractal},
but the connection to simulability transitions, or the relation to entanglement transitions is not yet understood.
\begin{figure}[t]
    \centering
        \includegraphics[width=8.2cm]{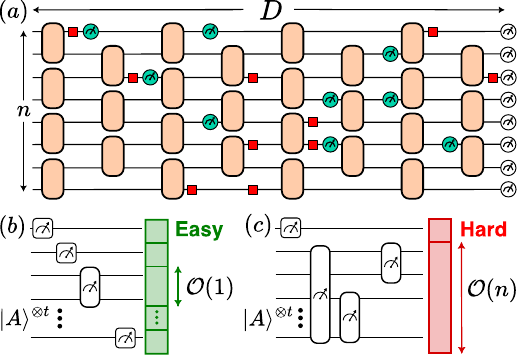}%
    \caption{
                $(a)$: 
                Monitored random Clifford$+T$ circuit with $n$ qubits and depth $D$. 
                Random 2-qubit Clifford gates (orange rectangles) form a brickwork architecture. 
                $T$ gates (red squares) and monitoring $Z$-measurements
                (green circles) are applied to individual qubits between layers of Cliffords with probabilities $q$ and $p$, respectively. 
                The circuit ends with a complete set of computational basis (i.e., $Z$-) measurements. 
                The aim is to sample from the output probability distribution.
                Panels $(b)$ and $(c)$ illustrate PBC: a $t$-qubit magic state register subject to mutually commuting Pauli measurements. 
                Each of these PBCs are equivalent to some circuit as in $(a)$, with $t$ $T$ gates.
                $(b)$: 
                In an easy phase, the measurements can be done in parallel on size $\mathcal{O}(1)$ blocks of magic states: %
                the magic remains fragmented after the measurements.
                $(c)$: 
                In a hard phase, most measurements belong to a size $\mathcal{O}(n)$ (for $t \propto n$) block: %
                the magic is diffused by the measurements.
                }
    \label{fig:circuit}
\end{figure}
In our work, we demonstrate a simulability phase transition driven by the dynamics of magic in the circuit, and show that this transition is related to but separate from the entanglement transition.
We consider a (1+1)D model involving random 2-qubit Clifford gates, non-Clifford $T$ gates, and single-qubit monitor $Z$-measurements, cf. Fig.~\ref{fig:circuit}($a$). 
Clifford gates along with the $T$ gate form a universal gate set for quantum computation~\cite{bravyi2005universal}. 
Hence, our model interpolates between classically simulable and universal circuits, controlled both by the level of $T$ gate doping and the rate of measurements.
This provides a toy model for the dynamics of quantum computers and other quantum systems which can be used to probe the barrier between high and low complexity regimes.

One can classically simulate a quantum circuit in time scaling exponentially only with the number of non-Clifford gates (injecting magic) and not with the number of qubits~\cite{Bravyi_PBC2016}.
However, as we shall note, for this exponential scaling (and hence computational hardness) to set in, locally injected magic must be able to spread in the system. 
To assess this in our Clifford+$T$ circuits, we use Pauli-based computation (PBC)~\cite{Bravyi_PBC2016}. 
This distills a Clifford circuit with $t$ $T$ gates into a $t$-qubit magic state register subject to mutually commuting measurements; this strips away all the classically efficiently simulable aspects, thus capturing the dynamics' true quantum essence.

We take $t$ to scale at least as the number $n$ of qubits; this allows for regimes with exp$(n)$ runtime for classically simulating PBC~\cite{Bravyi_PBC2016} (hard phase).
However, as we shall show, measurements may reduce this runtime to poly$(n)$ (easy phase), with a ``magic transition" at a critical monitoring rate.
In the easy phase, measurements fragment the magic state register into pieces whose size does not scale with $n$, cf. Fig.~\ref{fig:circuit}$(b)$. 
This fragmentation can be linked to the spread of magic in the original circuit (which we capture with a ``causal cone of magic")
becoming arrested by a mechanism we dub dynamical ``stabilizer purification", where sufficiently frequent measurements keep projecting the system into a stabilizer state. 
We show that entanglement and magic transitions may coincide, but also show that the latter can, strikingly, also occur in a volume-law phase. 
This shows that changes in the dynamics of magic alone, without a change in entanglement, can drive simulability transitions.
\section{Summary of the main results \label{sec:summary}}
Before providing our detailed analysis, we summarize our main results and the structure of the paper.
\subsection{Simulability transition, stabilizer-purification and magic fragmentation}\label{sec:SP_MF}
Here, we outline dynamical magic transitions and describe stabilizer purification and magic fragmentation.
The transition is introduced in Sec.~\ref{sec:CPXtrans} in detail and its mechanism is thoroughly discussed in Sec.~\ref{sec:SPandCartoon}. 

\setlength{\skip\footins}{0.5cm}
We study a model of random Clifford gates interspersed with random monitoring $Z$-measuremenents and non-Clifford $T$ gates, shown in Fig.~\ref{fig:circuit}($a$). 
Such circuits are generically hard to classically simulate, which we here define as weak simulation:\footnote{PBC can also perform ``strong" simulation, i.e., calculate a specific outcome's probability.} sampling from the distribution of computational basis measurements in their final state.
Nevertheless, we find that a certain ``runtime proxy" for classically simulating these circuits via PBC (see Sec.~\ref{sec:model_sim}) undergoes a transition at a critical monitoring rate: 
below this rate the circuits are hard to classically simulate using  PBC, while above it they become easy to simulate. 

PBC produces a simplified circuit that is equivalent, up to efficient classical processing, to the original circuit.
The PBC circuit acts on a ``magic state register" (MSR), see Fig.~\ref{fig:circuit}$(b)$,$(c)$, with each magic state stemming from implementing a $T$ gate in the original circuit. 
PBC then involves performing a series of mutually commuting measurements on the MSR. 
To classically simulate MSR measurements, one decomposes the initial product of magic states into a superposition of stabilizer states and simulates a Clifford circuit for each stabilizer state. 
The number of stabilizer states entering the superposition is called the \textit{stabilizer rank} of the state; for $t$ magic states this is believed to scale as $2^{\alpha t}$, where $\alpha>0$ is a constant~\cite{Bravyi_PBC2016, bravyi2016improved, Qassim2021improvedupperbounds}. 
We choose the parameters of our model such that the expected value of $t$---and hence the size of the MSR---scales as $\mathrm{poly}(n)$.

The supports of the MSR measurements are crucial for the definition of the simulability proxy for the magic transition.
Without monitoring measurements, the magic injected by the $T$ gates generically spreads and this leads to MSR measurement operators developing support on a large fraction of the MSR [Fig.~\ref{fig:circuit}$(c)$, Sec.~\ref{sec:CpxProxy}, \ref{sec:CCM}, and App.~\ref{app:CCM}].
In this case, simulation is hard because one has to consider at least $2^{\alpha \text{poly}(n)}$ stabilizer states.
Conversely, mechanisms that make the MSR measurements local may allow for simulating local ($n$-independent-sized) blocks of the MSR separately [Fig.~\ref{fig:circuit}$(b)$], leading to  easy simulation, since there are $\mathrm{poly}(n)$ blocks altogether. 
In this case, when the MSR measurements can be separated into disjoint $\mathcal{O}(1)$-sized MSR blocks we say \textit{magic fragmentation} (MF) occurs (or, more precisely, persists, since the MSR started out as fragmented before the measurements).
The MSR measurements thus either fragment or diffuse the initial magic in the MSR, and the presence of MF suggests that the spread of the magic inserted by the corresponding $T$ gates became arrested in the original circuit. 
A key mechanism for the dynamics of magic is the competition between $T$ gates and measurements. 
In the original circuit, $T$ gates tend to increase the stabilizer rank of the time-evolved state, while monitors tend to decrease it by projecting single qubits onto a stabilizer state. 
From the PBC perspective, the $T$ gates increase the MSR, while monitors tend to both reduce the MSR stabilizer rank and localize the supports of MSR measurements. 
We present numerical evidence for the existence of and transition between the corresponding two MSR regimes in Sec.~\ref{sec:CPXtrans}.

Measurements may eventually lead to the collapse of a complex superposition of stabilizer states to a single one. 
We call this \textit{dynamical stabilizer-purification} (SP) owing to its similarity to dynamical purification~\cite{gullans2020purification}.
SP defines a mechanism for arresting the spread of magic, and in PBC, we shall indeed show  that it can cause MF, and thus, can drive the magic transition.
In Sec.~\ref{sec:SPandCartoon}, we exemplify this via a simplified model where $T$ gates are temporally-separated enough such that the monitors after a $T$ gate project onto a stabilizer state before the next $T$ gate occurs; this is an example of SP.
Converting this to PBC, we find that the magic states for the successive $T$ gates are each subjected to their own single-qubit measurement. 
Thus, in this case, we see how SP in the original circuit leads to MF in PBC.
Therefore, it is appealing to search for regimes where the assumptions of the simplified model are met since these regimes would reveal easy phases.
We outline two such regimes in the following subsections.
\subsection{Uncorrelated monitoring \label{sec:AgMon}}
Here, we summarize how the SP probability is set by the EE in a model where monitors are sampled independently from $T$ gates. 
We overview the implications for a simulability transition
and provide some regimes where the magic and entanglement transitions coincide or differ.

As we shall show, the entanglement can set the probability to stabilizer-purify:
For stabilizer states, volume- or area-law scaling of the entanglement entropy implies most stabilizer generators are delocalized or localized, respectively (cf. App.~\ref{app:bulkMons}). 
Using this, we shall show that a single $T$ gate is stabilizer-purified with high probability in $\exp (n)$ time or $\mathcal{O}(1)$ time in the volume- and area-law phase, respectively (see Sec.~\ref{sec:SPTime_purifier} and App.~\ref{app:bulkMons}).

This will allow us to show that the entanglement and the magic transitions can coincide.
We focus on a regime with one $T$ gate occurring in every $\mathrm{poly}(n)$ time-steps (i.e., circuit layers). 
In this case, in the area-law phase, as the magic from each $T$ gate stabilizer-purifies in $\mathcal{O}(1)$ time, the $T$ gates are sufficiently far apart to enable SP one $T$ gate at a time. 
The area law thus implies easy PBC simulation. 
In contrast, in the volume-law phase, there are not enough monitors to SP: $T$ gates occur every $\sim \mathrm{poly}(n)$ time-steps but each requires $\exp (n)$ time to stabilizer-purify; hence, we expect a hard PBC phase since the stabilizer rank in the original circuit blows up, magic spreads (as we argue in Sec.~\ref{sec:SPTgate} and App.~\ref{app:CCM}), and the MSR measurements are delocalized. 
We find remarkable agreement between these expectations (see Sec.~\ref{sec:SPtime_implications}) and numerical simulations of the runtime proxy (see Sec.~\ref{sec:CPXtrans}), which confirm the link between entanglement and magic transitions in this regime. 

However, there are other regimes where the transitions are distinct. 
As we shall show, the area law can also set in without PBC becoming easy.
Consider the regime with $\mathcal{O}(n)$ $T$ gates injected into the circuit at every time step. 
This introduces $T$ gates at a rate higher than that at which each of them is stabilizer-purified, making SP exponentially unlikely. 
Thus, the stabilizer rank in the original circuit blows up and the MSR measurements are delocalized, so the PBC simulation is hard, regardless of the EE. 
Using a mapping to percolation (cf. Sec.~\ref{sec:perc}, App.~\ref{app:SpacetimePart}), we show that the simulation by PBC does eventually become easy, but this is because the (final-state relevant parts of the) circuit itself gets effectively disconnected, and this happens well after area-law EE sets in~\cite{skinner2019mipt}. 
We perform numerical simulations of the runtime proxy that agree well with these expectations (see Sec.~\ref{sec:qCPX}).
While these results already show that the relation between entanglement and magic transitions can be subtle, they might suggest that classical simulability in monitored Clifford+$T$ circuits depends only on entanglement: in the regimes discussed, PBC was hard in the volume-law phase and once an area law sets in, one can use MPSs (regardless of the hardness or ease of PBC) for efficient classical simulations. 
As we next discuss, however, magic transitions can also happen in the volume-law phase.

\subsection{\texorpdfstring{$\bm{T}$}{}-correlated monitoring \label{sec:InfMon}}
To push the magic transition into the volume-law phase, we introduce correlations between $T$ gates and monitoring measurements. 
It is amusing to interpret this scenario as there being a monitoring observer whose aim is to make classical simulation as easy as possible by measuring a fraction of the qubits
per circuit layer, and who may have some (potentially limited) knowledge of the locations of the $T$ gates.
Their best strategy is to attempt to measure immediately after $T$ gates. 
If we consider the extreme case with the observer having as many monitors as $T$ gates %
and perfect knowledge of where $T$ gates occurred, 
then the state is %
stabilizer-purified after each layer.
Indeed, this is an ``easy" point in parameter space regardless of the amount of entanglement and even for $\mathcal{O}(n)$ $T$ gates injected per layer. %
Remarkably, %
an entire easy phase can emerge within a volume-law entangled phase.

To turn this into a simulability transition within the volume-law phase, we can use as a control knob the information the observer has about where the $T$ gates are in spacetime.
Mathematically, this amounts to using the conditional probability of applying a monitor, given that a $T$ gate was present, as a control parameter. 
Focusing within the volume-law phase and in the regime where $\mathcal{O}(n)$ $T$ gates are applied per layer, we expect to find a phase transition from the previously found hard phase (cf. Sec.~\ref{sec:AgMon},~\ref{sec:qCPX}, corresponding to a zero-knowledge observer) to an easy phase as the observer's knowledge increases. 
We provide numerical simulations of the runtime proxy %
in Sec.~\ref{sec:MonGame}, which agree well with this expectation.
\subsection{Outline}
The rest of the paper is organized as follows.
In Sec.~\ref{sec:model_sim}, we present our circuit model, review PBC, and comment on the link between the dynamics of magic and the weights of PBC measurements. We then introduce a runtime proxy for classically simulating PBC, and define a corresponding order parameter. We also outline how a mapping to percolation can identify blocks in the original circuit amenable for separate simulation.
In Sec.~\ref{sec:CPXtrans}, we provide numerical evidence that a magic transition exists at a critical monitoring rate.
In Sec.~\ref{sec:SPandCartoon}, we propose SP as a mechanism for this transition, showing how it removes magic from the circuit and how this relates to MF in PBC. 
In Sec.~\ref{sec:MonGame}, we introduce the $T$-correlated monitoring model where a dynamical magic transition can occur within a volume-law phase. 
Finally, in Sec.~\ref{sec:Discussion}, we discuss our findings and future directions.
\section{Quantum circuit model and its simulation \label{sec:model_sim}}
\subsection{Monitored Clifford\texorpdfstring{$\bm{+T}$}{} circuits \label{sec:models}}
We shall consider Clifford$+T$ quantum circuits acting on $n$ qubits and with depth $D$, taking  $D = \mathrm{poly}(n)$.
The circuit architecture is shown in Fig.~\ref{fig:circuit}$(a)$.   
The Clifford gates are 2-qubit gates in a brickwork pattern; each gate is chosen randomly from a uniform distribution over the 2-qubit Clifford group $\mathcal{C}_2$~\cite{nielsen_chuang_2010}. 
Between each Clifford layer, we randomly apply the non-Clifford $T = \text{diag}(1, e^{i\pi/4})\sim e^{-i\pi Z/8}$ gate to each qubit with probability $q$ (which may be a function of $n$).
We also apply projective monitoring $Z$ measurements to certain qubits between the Clifford layers. 
(We could alternatively perform $X$ or $Y$ monitor measurements; we do not expect this would change the results obtained.)
Monitoring not only alters the state of the system but, by retaining measurement outcomes, it maintains a record of the state's evolution.
The last step of the circuit is a complete set of computational basis measurements.
We are interested in the difficulty of simulating these final measurements, i.e.,
sampling from the circuit's output distribution.

We consider two different models for the way in which qubits are monitored.
In the first, each qubit is measured in between Clifford layers with probability $p$ [Fig.~\ref{fig:circuit}$(a)$] irrespective of the preceding $T$ gates.
We term this the ``uncorrelated monitoring model."
In the second, we consider correlated $T$ gates and measurements, which can also be viewed as a monitoring observer who performs $pn$ measurements per circuit layer at locations of their choosing.
The aim of this observer is to make the final computation as simple as possible.
In the uncorrelated monitoring model, the observer is unaware of the locations of the $T$ gates, whereas in the alternative ``$T$-correlated monitoring model", they have this information and can use it to facilitate their task.
In this latter model, the spacetime locations of monitors and $T$ gates become correlated.

We are interested in how $p$ and $q$ influence the hardness of classically sampling from the output probability distribution, i.e., weak simulation~\cite{QRS23}.
Our main focus is the role of magic, thus the deviation from stabilizer-simulability; we assess this by developing a classical runtime estimate $\mathrm{\textsf{CPX}_{\mathrm{PBC}}}$ (cf. Sec.~\ref{sec:CpxProxy}) for exactly weakly simulating a typical quantum circuit using PBC. 
Our primary interest is whether $\mathrm{\textsf{CPX}_{\mathrm{PBC}}}$ scales exponentially (\textit{hard} phase) or polynomially (\textit{easy} phase) with $n$. 
Next, we review PBC and how it enables us to distill the essential quantum core of the simulation.

\subsection{Pauli-Based computation, magic spreading, and the runtime proxy}\label{sec:CpxProxy}
For Clifford$ + T$ circuits, PBC provides a natural method for classical simulation~\cite{Bravyi_PBC2016}. 
PBC is a quantum computational model that can efficiently sample from the output distribution of a quantum circuit involving poly$(n)$ Cliffords and $t=\text{poly}(n)$ $T$ gates.
It requires only the ability to do up to $t$ commuting Pauli measurements on $t$ qubits, prepared in a suitable initial state, along with poly$(n)$ classical processing. 
When performing quantum computation, PBC thus distills the essentially quantum parts of the problem and uses a quantum processor for these, while offloading the classically efficiently doable parts to a classical computer.
By the same logic, PBC also provides a route to classically simulating a circuit in a time that scales only as $\exp (t)$ rather than $\exp (n)$, a considerable advantage for circuits dominated by Clifford gates. 
Here, we summarize the process of converting a Clifford$+T$ circuit into PBC and the classical simulation of the latter. 
\begin{figure*}[t]
    \centering
        \includegraphics[width=17.2cm]{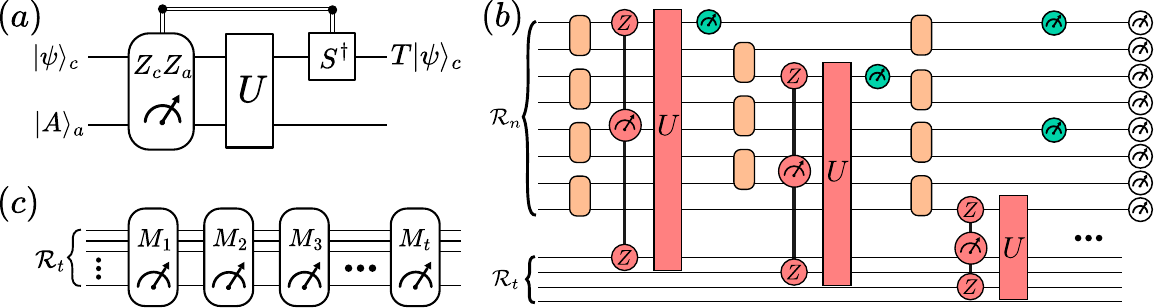}%
    \caption{%
    Pauli-based computation (PBC) for the quantum circuit in Fig.~\ref{fig:circuit}, cf. Sec.~\ref{sec:models}.
    $(a)$: Magic state gadget applying a $T$ gate to  qubit $c$ in state $\ket{\psi}_c$ while consuming ancilla qubit $a$ in state $\ket{A}_a$. 
    The gadget involves measuring $Z_cZ_a$, followed by applying $U = \exp \left( -i \pi Z_c X_a /4 \right)$ 
    and, depending on the measurement outcome, applying $S^\dagger$. 
    $(b)$: The monitored circuit from Fig.~\ref{fig:circuit} (acting on the computational register $\mathcal{R}_n$) but with magic state gadgets, shaded red, replacing $T$ gates (the circuit up to the first three $T$ gates is shown).
    $\mathcal{R}_t$ is the register of magic states. Gadget measurements are pre-selected to $-1$ for simplicity, so no adaptive $S^\dagger$ gates are required. 
    $(c)$: The PBC
    resulting from the circuit in $(b)$ is a series of commuting measurements acting only on $\mathcal{R}_t$.}
    \label{fig:PBC_fig}
\end{figure*}
To convert a monitored quantum circuit [including final measurements, cf. Fig.~\ref{fig:circuit}$(a)$] into a PBC circuit, one begins by replacing all $T$ gates with ``magic state gadgets''~\cite{magic0}, 
simultaneously introducing $t$ ancilla qubits, each in a so-called magic state: %
\begin{align}
\ket{A} = \frac{1}{\sqrt{2}} \left( \ket{0} + e^{i\pi/4} \ket{1} \right).
\end{align}
The magic state gadget is a procedure involving only Clifford gates and measurements, acting on the target qubit and a single ancilla magic state.
This gadget is shown in Fig.~\ref{fig:PBC_fig}$(a)$.
It involves the measurement of the joint parity operator $Z_c Z_a$, where $Z_c$ is the Pauli operator acting on the target (computational) qubit, and $Z_a$ is that acting on the ancilla qubit.
The measurement is then followed by the action of the following Clifford gate:
\begin{align}\label{eqn:Gadget_Clifford_U}
    U = \exp{\left(-i\frac{\pi}{4}Z_c X_a\right)}\in\mathcal{C}_2,
\end{align}
where $\mathcal{C}_k$ is the $k$-qubit Clifford group.
If the measurement outcome is $+1$, then $S^\dagger$ is next applied to the target qubit. 

After replacing all $t$ $T$ gates with the above gadgets, we have two registers: an $n$-qubit computational register and a $t$-qubit MSR.
The circuit acting on these two registers now involves Clifford gates, gadget measurements (GMs), monitoring measurements and final ``output" measurements, see Fig.~\ref{fig:PBC_fig}$(b)$ for an example. Our simulation goal is to sample the output measurement outcomes.

We further simplify this circuit by commuting each Clifford gate past all measurements, which are thus updated $M \mapsto C^\dagger M C$ for Clifford gate $C$ and original measurement operator $M$. 
Once the Clifford gate is commuted past all measurement operators, it no longer affects the final output distribution and so can be deleted. 
The set of measurements can then be restricted, with only poly$(n)$-time classical processing time~\cite{Bravyi_PBC2016,Yoganathan2019QuantumAO} (see App.~\ref{app:PBC_Details} for review), to a mutually commuting subset that acts non-trivially only on the MSR.
Therefore, the original computational register can be deleted and we are left with at most $t$ commuting measurements performed on $t$ magic states.
The simulation of the initial circuit can be replaced by the simulation of these MSR measurements [and the poly$(n)$-time classical processing], cf. Fig.~\ref{fig:PBC_fig}$(c)$.

With PBC defined, we can now link the easy and hard phases in Fig.~\ref{fig:circuit}$(b),(c)$ to the dynamics of magic in the original circuit. 
We first note that the hard regime requires magic to spread; in other words, operator spreading is a requirement for multi-qubit MSR measurements to build up.
This could be due to the spreading of $T$ gates, supplying an intuitive path for magic spreading; 
but this can also be combined with the spreading of stabilizers preceding a $T$ gate since nonlocal stabilizers can allow for a $T$ gate's magic to instantly delocalize (see Sec.~\ref{sec:CCM} for a further discussion).
For a monitor or output measurement to couple to  multiple $T$ gates, and hence be able to yield multi-qubit MSR measurements, it must be in the causal future of the magic injected by these $T$ gates, as detailed in Sec.~\ref{sec:CCM}.
Since the $T$ gates occur randomly, coupling $\mathcal{O}(n)$ of them to a measurement requires spatial operator spreading, and thus locally inserted magic to spread.
The role of monitors is to interrupt the spreading of magic (and thus the creation of increasingly complex stabilizer superpositions). This is the essence of how SP can arrest the dynamics of magic and lead to MF, as we shall see in Sec.~\ref{sec:SPandCartoon}.

To characterize this competition with an order parameter, we must first describe how one classically simulates PBC. This involves evaluating the probabilities $P(\mathbf{m})$ 
for measurement outcomes %
$\mathbf{m}~=~(m_1, \ldots , m_k)$ %
at any time step $k\leq t$.
By computing $P(m_1)$, then $P(m_2|m_1) = P(m_2,m_1)/P(m_1)$ etc. one can flip coins with appropriate biases to simulate the PBC measurements~\cite{Bravyi_PBC2016}.
Therefore, we wish to evaluate  $\bra{A}^{\otimes t} \Pi_{\mathbf{m}} \ket{A}^{\otimes t}$ for  $\Pi_{\mathbf{m}} = \prod_{j=1}^k \frac{1}{2}(1 + m_j M_j)$, where the $M_j$ are the commuting PBC measurements. 
By decomposing $\ket{A}^{\otimes t}$ into a low-rank sum of $\chi_t$ stabilizer states, one can perform these evaluations in $\mathcal{O}(t^3 \chi_t^2) =  \text{poly}(t) 2^{\mathcal{O}(t)}$ time~\cite{Bravyi_PBC2016,bravyi2016improved,Qassim2021improvedupperbounds}.

If the MSR can be partitioned (as is the case for MF), this can speed up and parallelize this classical computation.
Let us suppose we split the measurements into the largest number $K$ of subsets $R_i$ ($i=1,\ldots,K$) with the constraint that no measurement in $R_i$ has support overlapping with the support of any measurement in $R_{j\neq i}$. 
That is, if we let $\text{Sup}(R_i)$ be the union of the supports of all measurements in $R_i$, then $\text{Sup}(R_i)$ are mutually disjoint sets.
Let $t_i = |\text{Sup}(R_i)|$ be the size of $\text{Sup}(R_i)$, where $\sum_i t_i \leq t$.
Now evaluating the measurement probability for time step $k$
involves evaluating a probability for each MSR partition $P(m_1,\ldots,m_k) = \prod_{i=1}^K P_i(k)$, where   
\begin{align}
P_i(k) = \bra{A}^{\otimes t_i} \prod_{\substack{j: \; j\leq k \\[0.3em] M_j\in R_i}} \frac{1+ m_j M_j}{2}\ket{A}^{\otimes t_i}.
\end{align}
Thus the runtime in evaluating the probability becomes
\begin{align}
\mathcal{O}\left( \sum_{i=1}^K |R_i| t_i^3 \chi_{t_i}^2 \right) = \sum_{i=1}^K \text{poly}(t_i) 2^{\mathcal{O}(t_i)}
\end{align}
since there are $|R_i|$ measurements in subset $i$, each of which acts on $t_i$ qubits.
This quantity is exponential only in the parameters $t_i$, not necessarily in $t$.
Note that if each $t_i = \mathcal{O}(1)$, corresponding to a magic fragmentation regime (cf. Sec.~\ref{sec:SP_MF}), then the entire computation can be executed in poly$(n)$ time.
Therefore, we define the following runtime proxy that captures the exponentially scaling part of the above runtime for simulating the PBC:
\begin{equation}
    \text{\textsf{CPX}}_{\text{PBC}} = \sum_{i=1}^{K'} 2^{t_i}. \label{eq:PBCcpx}
\end{equation} 
Here, we restrict the sum to the $K' \leq K$ MSR partitions $R_i$ that support at least one of the final output measurements since computing the probabilities for gadget measurements is trivial and we do not need to calculate the monitor measurement probabilities since their outcomes are given (cf. App.~\ref{app:PBC_Details}).
This runtime proxy differs from the actual runtime by poly$(n)$ prefactors and also by $\mathcal{O}(1)$ factors in the exponents.
However, we are interested only in the efficiency of PBC-based classical simulation, i.e., whether the runtime scales polynomially or (at least) exponentially with $n$.
This is indeed captured by the scaling of \textsf{CPX}$_\text{PBC}$.

We define the \textit{order parameter} in terms of the \textit{typical} value of $\textsf{CPX}_\mathrm{PBC}$ among random circuit realizations:
\begin{align}\label{eq:order_param}
    \log \mathrm{\textsf{CPX}}^{(\mathrm{typ})}_{\mathrm{PBC}}/t \equiv \mathbb{E}_\text{RQC} \left[ \log (\mathrm{\textsf{CPX}}_{\mathrm{PBC}})/t \right],
\end{align}%
where $\mathbb{E}_\text{RQC}$ is the expectation value over the uniformly distributed Clifford gates and the randomly placed monitor measurements and $T$ gates.
In a hard phase (no MF), we expect \textsf{CPX}$_\text{PBC}^\text{(typ)} = \exp(t)$, thus, the order parameter would be $\mathcal{O}(1)$ and positive.
In an easy phase (with MF), we expect \textsf{CPX}$_\text{PBC}^\text{(typ)} = \mathrm{poly}(n)$, hence, the order parameter would vanish as $n\to\infty$ since $t = \mathrm{poly}(n)$.

We emphasize that, by measuring the ``unfragmented" magic in PBC, Eq.~\eqref{eq:order_param} accounts for the fraction of the potentially present magic ($t$) structured such that it results in classically hard to simulate PBC (as $n$ is increased). In terms of magic dynamics, unlike existing global magic measures, one can view Eq.~\eqref{eq:order_param} as a proxy for the fraction of injected magic that can spread in the original circuit.

\subsection{Circuit cluster selection and percolation \label{sec:perc}}
For frequent enough monitors, only a small part of the entire circuit history suffices for simulation.
Projectively measuring a qubit makes the previous state partially irrelevant for simulation: as an extreme case, if all qubits are monitored at the same time-step $t_0$, then determining the final state requires simulating only the evolution after $t_0$.
Thus, monitors disconnect the circuit temporally.
Similarly, separable Clifford gates (i.e., those 2-qubit gates of the form $C_1 \otimes C_2$ for $C_1,C_2\in \mathcal{C}_1$) disconnect the circuit spatially by allowing for the simulation of neighbouring qubits in parallel, cf. App.~\ref{app:stp}.
Therefore, we can simplify the simulation by mapping our circuit architecture from Fig.~\ref{fig:circuit}$(a)$ to a percolation model (cf. App.~\ref{app:mapping}),  and focusing on circuit clusters connected to the final-time boundary.
The numerical experiments for computing \textsf{CPX}$_\text{PBC}$ from Sec.~\ref{sec:CPXtrans},~\ref{sec:qCPX},~\ref{sec:MonGame} use this optimization procedure, i.e., they apply the PBC procedure only on the relevant circuit clusters.

This percolation model sets an upper bound for the critical monitoring rate of a simulability transition (cf. App.~\ref{app:TNtransition}).
The size and depth of the selected circuit clusters directly sets the runtime for using exact tensor network (TN) contraction to sample from the output distribution~\cite{markov2008sim}.
Above a critical monitoring rate $p_\mathrm{c}^\mathrm{TN} \simeq 0.48$, the size and depth of the clusters are $\mathcal{O}(1)$; thus classically simulating the circuit by TN contraction is easy.
Indeed, the finite size of the clusters also results in the PBC method being efficient for any value of $q$.
Therefore, $p_\mathrm{c}^\mathrm{TN}$ sets an upper bound for a simulability phase transition.
In Sec.~\ref{sec:CPXtrans} and~\ref{sec:MonGame}, we shall study settings where such a transition occurs at $p_\mathrm{c} < p_\mathrm{c}^\mathrm{TN}$, while in Sec.~\ref{sec:qCPX}, we describe settings where the bound is saturated.

Note that our use of percolation theory differs from Ref.~\onlinecite{qcpx_percolation}, which focuses on state complexity. As Ref.~\onlinecite{qcpx_percolation} notes, this can be distinct from the complexity of weak simulation, 
which our percolation mapping (accounting for separability in $\mathcal{C}_2$) assesses. 
\section{Magic Transitions with Uncorrelated Monitoring}
\label{sec:CPXtrans}
Before discussing, in Sec.~\ref{sec:SPandCartoon}, the mechanism behind the dynamical magic phase transitions we study, we first describe the $\mathrm{\textsf{CPX}_{\mathrm{PBC}}}$ transition that occurs in the uncorrelated monitoring model: 
we show results from numerical simulations of the runtime proxy $\mathrm{\textsf{CPX}_{\mathrm{PBC}}}$, for fixed $qD=\mathcal{O}(1)$ [implying $t=\mathcal{O}(n)$], and show that MF is linked with the magic transition.
We describe the setup of the numerical experiment, discuss the results, and highlight where the mechanism based on SP from the next section will come into play.
\begin{figure}[htpb]
    \centering
        \includegraphics[width=8.6cm]{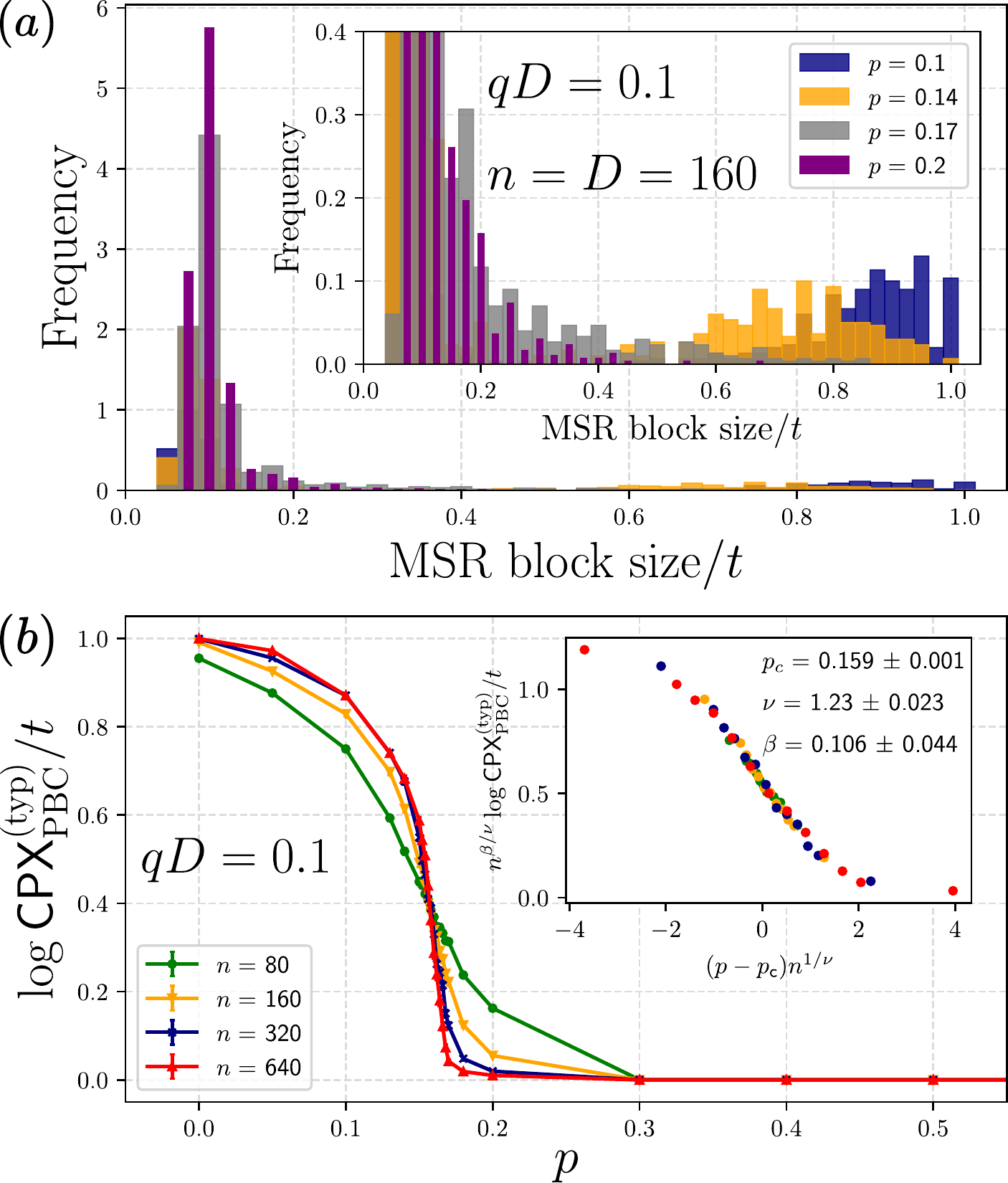}%
    \caption{
                Dynamical magic phase transition for fixed $qD = \mathcal{O}(1)$ coinciding with the entanglement phase transition. 
                $(a)$: The MSR is fragmented into $\mathcal{O}(1)$ blocks for $p>p_\mathrm{c} \approx 0.159$ while it has an $\mathcal{O}(n)$ size block
                for $p<p_\mathrm{c}$. 
                The inset more closely shows the histogram of block sizes averaged over 300 realizations.
                $(b)$: The order parameter $\log \mathrm{\textsf{CPX}}^{(\mathrm{typ})}_{\mathrm{PBC}}/t$ versus measurement probability $p$. The error bars, showing the standard error of the mean (SE), are imperceptible. 
                The  order parameter drops from a finite value for $p<p_\mathrm{c}$ to zero for $p>p_\mathrm{c}$ (it remains zero up to $p=1$ but this is not shown).
                The inset shows the finite-size scaling collapse characteristic of a continuous magic phase transition.
            }
    \label{fig:qD_CPX}
\end{figure}
The original circuits  in our setups, subsequently expressed as PBC, follow Fig.~\ref{fig:circuit}$(a)$. 
We perform the classical part of PBC, as detailed in App.~\ref{app:PBC_Numerics}, to find the set of PBC measurements. 
From this set, we infer the sizes of the MSR blocks and the runtime proxy $\mathrm{\textsf{CPX}_{\mathrm{PBC}}}$. 
For concreteness, we take $D=n$, however, any $D = \mathrm{poly}(n)$ depth circuit should generically yield the same results.
We find a MF-driven magic phase transition at $p_\mathrm{c} \approx 0.159$, a value consistent with a simultaneous entanglement transition~\cite{li2019mipt,zabalo2020EEpc}.  
Below $p_\mathrm{c}$ (hard phase),
the distribution of MSR block sizes is (shallowly) peaked at a value set by the total MSR size $t \sim qDn$, cf. Fig.~\ref{fig:qD_CPX}$(a)$, inset. 
Above $p_\mathrm{c}$ (easy phase), this distribution becomes peaked at unit block size, with a tail decaying to zero (faster for larger $p$) in an $n$-independent manner, cf. Fig.~\ref{fig:qD_CPX}$(a)$. 
The order parameter is non-zero for $p<p_\mathrm{c}$ and vanishes for $p>p_\mathrm{c}$, cf. Fig.~\ref{fig:qD_CPX}$(b)$. 
Using finite-size scaling~\cite{cardyFSS,andreas_sorge_2015_35293}, we find $p_\mathrm{c} = 0.159 \pm 0.001 $. 
The inset of Fig.~\ref{fig:qD_CPX}$(b)$ displays the corresponding scaling collapse, signaling a continuous phase transition~\cite{PhysRevLett.28.1516, cardyFSS} as a function of $p$. 
These results illustrate, firstly, that MF occurs upon increasing $p$, driving the system to the  ``easy" phase with vanishing order parameter (for $n\to\infty$).  
Secondly, noting that $p_\mathrm{c}$ is consistent with the critical value $p_\mathrm{c}^\mathrm{EE}=0.154\pm 0.004$ for the entanglement transition in Clifford circuits~\cite{li2019mipt,zabalo2020EEpc} (with the slight deviation possibly attributable to the $t=\mathcal{O}(n)$  $T$ gates)
we see that magic and entanglement transitions can co-occur.
As we shall argue, this co-occurrence is due to the tight link, for $qD = \mathcal{O}(1)$ with uncorrelated monitoring, between SP and entanglement. 
We studied other $qD = \mathcal{O}(1)$ values, finding similar results (see App.~\ref{app:add_qD}).
The key concept for understanding the appearance of MF and the coincidence of $\mathrm{\textsf{CPX}_{\mathrm{PBC}}}$ and EE transitions is SP. 
We shall see in the next section that SP causes MF and drives the simulability transition and that, for fixed $qD=\mathcal{O}(1)$, an area-law entanglement phase leads to SP.

\section{Magic Transitions via Stabilizer Purification\label{sec:SPandCartoon}}
Here, we present SP as a mechanism behind the magic transition described in the previous section.
The overall picture of this mechanism is: in an easy phase, each $T$ gate has its magic annihilated before the next one is applied; whereas in a hard phase, magic from many $T$ gates accumulates.
Our aim is to calculate the probability that the magic introduced by a $T$ gate can be removed by monitors before further $T$ gates would be applied. %
If this probability is high, then the evolved state is constantly stabilizer-purified, leading to an easy phase.
This probability depends on the temporal separation of $T$ gates and the monitoring probability.

We compute this probability by starting from a simplified model, then progressively refining it until it captures most features of the uncorrelated monitoring model. 
In Sec.~\ref{sec:SPTgate}, we study how a single monitor removes the magic of a single $T$ gate (or reduces the magic from a small number of $T$ gates) in both the original and the PBC circuits. %
In Sec.~\ref{sec:SPtoMF}, relating these two perspectives, we argue that SP leads to MF and hence a magic transition. 
In Sec.~\ref{sec:SPTime_purifier}, %
we calculate the probability to stabilizer-purify and explain how it relates to the entanglement phases.
In Sec.~\ref{sec:SPtime_implications}, based on the SP probability, we interpret the magic transition for uncorrelated monitoring described in Sec.~\ref{sec:CPXtrans}.
In Sec.~\ref{sec:SPdirect}, we describe how SP implies that stabilizer simulation of the original circuit (rather than the PBC circuit) is also easy.

\subsection{Stabilizer-purified \texorpdfstring{$\bm{T}$}{} gate \label{sec:SPTgate}}
Here, we explore the constraints a monitor has to satisfy to remove the magic introduced by a single $T$ gate. 
We assume that the system had been in a stabilizer state $\ket{\psi}$ before the $T$ gate acted and that $T\ket{\psi}$ is a non-stabilizer state, i.e., that $T$ splits $\ket{\psi}$ into a superposition of two stabilizer states.
The monitor is a Pauli measurement $M$ (absorbing in it the Cliffords between $T$ and the monitor), with projector $\Pi$ and study the conditions for $\Pi T \ket{\psi}$ being an (unnormalized) stabilizer state, i.e., for SP.
More details on the following considerations are in App.~\ref{app:SPTgate}.

In this first simplified model,
we shall consider $\ket{\psi}$ as a random stabilizer state on $n$ qubits, with stabilizer group $\mathcal{S} = \langle s_1, \dots, s_n \rangle $, where $s_i$ for $i=1,\dots,n$ are a complete set of stabilizer generators, uniquely specifying $\ket{\psi}$ by $s_i\ket{\psi}=\ket{\psi}$.
Without loss of generality, we take $T$ to act on the first qubit, $T=T_1$, and that $Z_1$ anti-commutes with $s_1$ (if it commuted with all $s_i$ then $T_1\ket{\psi}$ would be a stabilizer state). 
Then $T_1 \ket{\psi} = c_+ \ket{\psi_+} + c_- \ket{\psi_-}$, where $\ket{\psi_\pm}$ are stabilized by groups $\mathcal{S}_\pm = \langle \pm Z_1, g_2, \dots, g_n \rangle$ respectively, where $g_i$ are updated generators (i.e., $g_i = s_i$ if $[Z_1, s_i]=0$ and $g_i = s_1 s_i$ if $\{ Z_1, s_i \} =0$).

The resulting state can be interpreted as an encoded state of a stabilizer quantum error-correcting code~\cite{gottesman1997stabilizer}, with stabilizer group $\mathcal{G} = \langle g_2,\ldots,g_n\rangle$ and logical operators $Z_1$ and $s_1$; these mutually anti-commute but commute with all $g_i$ for $i=2,\ldots, n$. %
The single logical qubit of this code is in a magic state.
The only way for $M$ to yield a stabilizer post-measurement state is to measure the state of the logical qubit.
That is, $M$ must belong to one of the following cosets (up to an irrelevant sign):
\begin{align}\label{eq:Stab_pur_possibs}
    Z_1 \mathcal{G}, \; s_1\mathcal{G}, \; is_1 Z_1 \mathcal{G}. 
\end{align}
We prove this in App.~\ref{app:MonForm}. 

Consider what happens in PBC in the three cases of Eq.~(\ref{eq:Stab_pur_possibs}). 
After replacing the $T$ gate with a gadget and commuting the gadget Cliffords past
$M$ (which already absorbed the circuit Cliffords between $T_1$ and the monitor), we are left with the following sequence of measurements:
\begin{enumerate}
    \item Gadget measurement (GM) of operator $Z_1 Z_a$, where $Z_a$ acts on the ancilla qubit.
    \item Updated monitor measurement $M'=U^{\prime\dagger} MU^\prime$ where
    $U^\prime=(S_1^\dagger)^{(1+m)/2}U$ with $m$ the gadget measurement outcome and $U$ the gadget Clifford from Eq.~(\ref{eqn:Gadget_Clifford_U}). Eq.~(\ref{eq:Stab_pur_possibs}) implies that $M'$ is in one of the following cosets (up to a sign): 
    \begin{align}\label{eqn:Monitor_mmt_possibs}
    Z_1 \mathcal{G}, \; s_1 X_a \mathcal{G},\; is_1 Z_1 X_a \mathcal{G}.
    \end{align}
\end{enumerate}
As noted, $Z_1$ anti-commutes with $s_1$, which stabilized the initial state  $\ket{\psi}$. 
Therefore, according to the PBC procedure, we replace the GM with a Clifford gate $V = \exp(\lambda\frac{\pi}{4} Z_1 Z_a s_1)$ for $\lambda = \pm 1$ chosen uniformly at random (see App.~\ref{app:PBC_Details}).
We then commute this $V$ past the updated monitor measurement.

In the first case from Eq.~(\ref{eqn:Monitor_mmt_possibs}), $V$ does not commute with the updated monitor and so updates it to an operator from coset $s_1 Z_a\mathcal{G}$ (up to a sign). 
In the second case, $V$ commutes with the updated monitor and hence leaves the measurement operator unchanged.
In the final case, $V$ does not commute with the measurement and so updates it to an operator from the coset $iX_a Z_a \mathcal{G}$ (up to a sign).

As can be seen, in all cases, the final result is that the monitor measurement (updated by the gadget Cliffords and by $V$) commutes with all members of $\mathcal{S}$, the stabilizer group of the initial state $\ket{\psi}$.
Therefore, it is retained in the PBC: it becomes a single-qubit measurement on the ancilla magic state, projecting it into a stabilizer state (see App.~\ref{app:MonForm} for more details).

\subsubsection{Causal cone of magic: linking magic spreading to stabilizer purification and  complexity\label{sec:CCM}}
\begin{figure}[t]
    \centering
        \includegraphics[width=0.8\columnwidth]{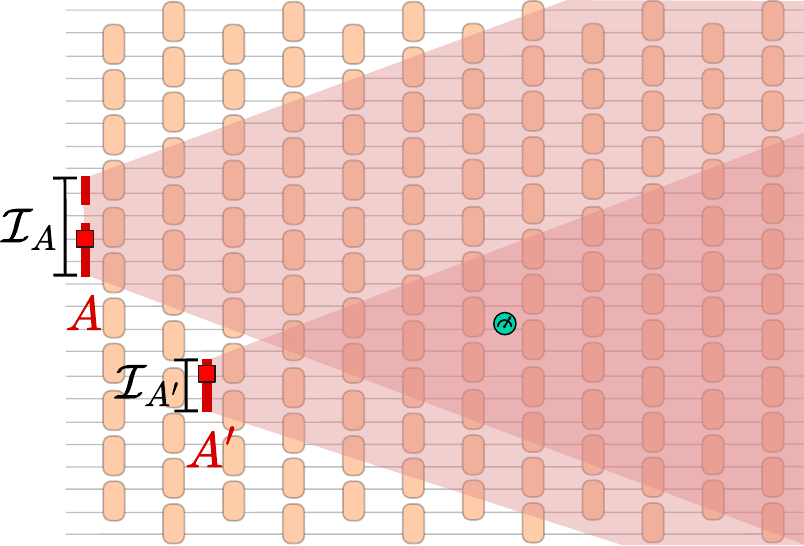}
    \caption{
            Causal cones of magic of two $T$ gates, extending from the apexes $\mathcal{I}_A$ and $\mathcal{I}_{A'}$ (black segments) for $A$ and $A'$ (bold red lines), respectively.
            A later measurement (green circle) can acquire support on the ancillas of both $T$ gates only if it is in the overlap of the causal cones.
            }
    \label{fig:CCM}
\end{figure}

Consider a circuit with
a single $T$ gate and the smallest set of spacetime locations outside of which a monitor cannot SP the $T$ gate for any choice of subsequent Clifford layers. 
While the detailed structure of this set depends on that of the logical operators, its convex hull forms a causal cone extending from the smallest interval $\mathcal{I}_A$ containing the union $A$ of the supports of minimum-width logical operators in $s_1 \mathcal{G}$, cf. Fig.~\ref{fig:CCM} (see also App.~\ref{app:CCM}).
We call this cone the \textit{causal cone of magic} (CCM) and $\mathcal{I}_A$ the cone's apex.
The $T$ gate's forward lightcone is always within this CCM, as intuitively expected. 
However, the CCM can be broader; it may have a nonzero-width apex since the magic from the $T$ gate can instantly delocalize and thus SP may occur from measurements elsewhere than the $T$ gate's support; in particular, $\ket{\psi_+}$ and $\ket{\psi_-}$ may be distinguished, and hence SP can occur from  monitors, throughout the support of the most local members of $s_1 \mathcal{G}$ (these toggle between $\ket{\psi_+}$ and $\ket{\psi_-}$) and it is this support that sets $A$. 
Measuring qubits outside the CCM however cannot SP. 
Thus such measurements preserve the magic injected by the $T$ gate. 
In this way, the CCM characterizes how magic from a $T$ gate can spread.

The CCM can also be used to link magic spreading to multi-qubit measurements in PBC and hence to  $\text{\textsf{CPX}}_{\text{PBC}}$. 
To illustrate this, consider two $T$ gates, $T_k$ and $T_l$, that both inject magic into the circuit, thus yielding two CCMs. 
The state $T_k T_l \ket{\psi}$ encodes two logical qubits, $k$ and $l$, associated through the gadgets to ancilla $a_k$ and $a_l$, respectively. 
We assume that there are no measurements in the circuit layers between the $T$ gates, for simplicity.
Consider a computational basis measurement (monitor or final output) in some layer after the second $T$ gate. 
For this measurement to be retained in PBC it must commute with all stabilizer generators and hence, to be nontrivial, it must measure a logical operator.
By the definition of the CCM, this can be a two-logical-qubit operator, thus coupling $a_k$ and $a_l$, only if it is in the intersection of both CCMs. 
(See App.~\ref{app:CCM} for further details.)
If there were monitor measurements, or such choices of Cliffords after the $T$ gates, that arrest the spread of magic from the apexes into the CCMs such that the logical operators' supports cannot develop an intersection, this would prevent the appearance of multi-qubit MSR measurements. 
Thus, the buildup of magic from both $T$ gates in the original circuit due to operator spreading is required for a buildup of the complexity in PBC. 

\subsection{SP leads to MF \label{sec:SPtoMF}}
Here, we outline how SP leads to MF, using a second simplified model building on the above picture.
We shall consider a circuit with an input stabilizer state acted on by poly$(n)$ circuit blocks.
Each of these blocks has poly$(n)$ depth, $\mathcal{O}(1)$ $T$ gates randomly placed between Clifford gates and projective measurements.
The main assumption of this second simplified model is that after each block, the state of the system is a stabilizer state.
Importantly, something akin to this picture resembles the easy-to-simulate phases in our models, cf. Secs.~\ref{sec:SPtime_implications},~\ref{sec:MonGame}.

A simple argument shows that stabilizer-purifying the $T$ gates from each block before the next one occurs yields size $\mathcal{O}(1)$ MSR measurements in the PBC.
First we prove the following:
\begin{theorem}
    \label{thm:PBC}
    The output of a monitored Clifford+$T$ circuit (acting on arbitrary initial stabilizer state) is a stabilizer state if and only if the output of the corresponding PBC is a stabilizer state.
\end{theorem}
\begin{proof}
Consider a magic measure $\mathcal{M}$ obeying the following properties: (i) $\mathcal{M}(\ket{\psi}) = 0$ if and only if $\ket{\psi}$ is a stabilizer state, (ii) $\mathcal{M}(C\ket{\psi}) = \mathcal{M}(\ket{\psi})$ for any Clifford gate $C$ and (iii) $\mathcal{M}(\ket{\psi}\otimes\ket{\phi}) = \mathcal{M}(\ket{\psi}) + \mathcal{M}(\ket{\phi})$.
Such a measure exists~\cite{magic7}.

We show that this measure is unchanged upon converting a Clifford+$T$ circuit with monitors into a PBC.
Observe that after replacing a $T$ gate with a gadget and applying that gadget to the target qubit, the result is that the target qubit has a $T$ gate applied to it and the ancilla qubit ends in an eigenstate of $Y_a$ (this can be seen by commuting the gadget Clifford $U$ before the gadget measurement). 
Hence, the state after application of a gadget is altered only by the addition of stabilizer states.
Therefore, the value of the magic measure is unchanged if we apply magic state gadgets instead of $T$ gates, owing to properties (i) and (iii) above.
After this, we commute all Cliffords past all measurements to the end of the circuit and delete them; using property (ii), this does not alter the magic of the final state.
We then go through the list of measurements and replace any measurement that anti-commutes with a previous one with a random Clifford gate, commuting that to the end of the circuit and deleting it as well. 
This replacement produces the corresponding post-measurement state of the replaced measurement (see App.~\ref{app:PBC_Details}) and hence does not change the magic in the system. 
Similarly to above, commuting this Clifford past remaining measurements and deleting it does not change the magic of the state either.
We can then restrict all measurements to the MSR without changing anything about the final state.
Since the computational register now remains untouched in its initial stabilizer state, deleting it does not affect the magic of the final state [properties (i) and (iii)].
This leaves only the PBC, after whose measurements on the MSR, the magic of the final state is the same as that of the final state of the original circuit.
Hence, if the output of the original circuit is a stabilizer state, so too is the output of the PBC, and vice versa [property (i)].
\end{proof}

Since each circuit block in the second simplified model is itself a monitored Clifford+$T$ circuit and it ends in a stabilizer state, the ancillas introduced in that circuit block all end up in a stabilizer state too, from Theorem~\ref{thm:PBC}.
After the first block, suppose $k$ ancillas have been introduced. Then the measurements introduced in the first block translate to PBC measurements that project those $k$ ancillas to a stabilizer state.
For the next block, we can view these $k$ states as part of the block's initial stabilizer state, thus reducing the effective MSR size for this block to $t-k$ where $t$ is the total number of $T$ gates in the initial circuit.
Thus all subsequent measurements of the PBC act trivially on the first $k$ ancilla qubits.

Proceeding in this way, the measurements from each circuit block correspond to PBC measurements that act trivially on all ancillas apart from those introduced within that block.
But because, by assumption, there are only $\mathcal{O}(1)$ of these ancillas introduced in each block, the measurements that project them into a stabilizer state must also have $\mathcal{O}(1)$ weight. 
That is, the SP of the original circuit corresponds to MF of the PBC.

\subsection{Stabilizer-purification probability and time 
            \label{sec:SPTime_purifier}}
Here,
with the aid of a third simplified model, we outline the calculation of the SP probability, link it to entanglement and introduce the SP time.
As above, we shall consider a circuit with poly$(n)$ circuit blocks, cf. Fig.~\ref{fig:TCB}.
However, now we assume each block has only one $T$ gate, at its start. The $T$ gate is followed by a depth-$d$ brickwork of $2$-qubit Cliffords with monitor $Z$-measurements between each Clifford layer occurring on each qubit with probability $p$. 
We dub this circuit block a $T$-\textit{circuit-block} (TCB).
The full circuit has as input a stabilizer state, then TCBs in succession. 
By $t=\mathrm{poly}(n)$, this model is a cartoon for our uncorrelated monitoring model for fixed $qD$; see also Sec.~\ref{sec:SP_AC}.
(Here, we do \emph{not} assume that the output of a TCB is a stabilizer state.)

SP is guaranteed if each of the TCBs purify the magic introduced by their $T$ gate, i.e., if they output a stabilizer state. 
If so, we can use the argument from Sec.~\ref{sec:SPtoMF} to show that MF occurs. 
To study when this applies, we shall estimate the \textit{stabilizer-purification time} $\tau_\text{SP}$, the characteristic depth such that for $d\gg\tau_\text{SP}$ a TCB is a $T$ \textit{purifier}, i.e., it almost surely purifies its $T$ gate, provided its input $\ket{\psi_\text{in}}$ is a stabilizer state.
The first step for this is finding the TCB's corresponding SP probability. 

\begin{figure}[t]
    \centering
        \includegraphics[width=8.6cm]{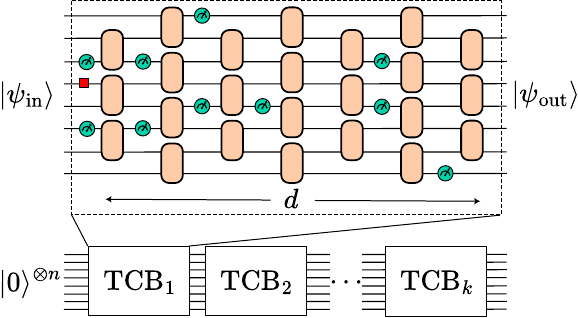}
        \caption{
                Illustration of the third and final simplified model introduced in this Section. 
                A depth-$d$ $T$-circuit-block (TCB) is illustrated (top). 
                It involves a $T$ gate (red box) being applied to input state $\ket{\psi_{\text{in}}}$, followed by random monitor measurements (green circles) and random 2-qubit Clifford gates.
                Measurements occur between Clifford layers on each qubit with probability $p$.
                The model involves $k$ TCBs applied to initial state $\ket{0}^{\otimes n}$ (bottom).
            }
    \label{fig:TCB}
\end{figure}
\subsubsection{Stabilizer-purification probability \label{sec:SPprob}}
We assume that $T\ket{\psi_\text{in}}$ is not a stabilizer state. 
Like in Sec.~\ref{sec:SPTgate}, the state encodes one logical qubit (in a magic state) with logical operators $Z_i$ and $s_1$, where $s_1$ stabilized $\ket{\psi_\text{in}}$ and the $T$ gate acts on qubit $i$ (see Sec.~\ref{sec:SPTgate}).
If bulk monitors SP, they measure the state of this logical qubit.
The SP probability of a TCB has contributions from monitors immediately after the $T$ gate, and from monitors in subsequent layers of the TCB,
\begin{equation}
    \mathbb{P} (\text{SP}) = \mathbb{P}(d^*=1) + \mathbb{P} (d \geq d^* > 1),
\end{equation}%
where $d^*\leq d$ is the depth at which SP occurs.
Here
\begin{align}
    \mathbb{P}(d^* = 1) =  1 - (1-p)^{w_0} \equiv p_1 \label{eq:p1def}
\end{align} 
for some number of qubits $w_0$ that, if measured, result in SP. $w_0\geq 1$ since monitors are $Z$-measurements and $Z_i$ is a logical operator [cf. Eq.~(\ref{eq:Stab_pur_possibs})] while $w_0\leq |A|$, with $|A|$ the cardinality of the set $A$ defining $\mathcal{I}_A$ (cf. Sec.~\ref{sec:CCM}) since outside of $A$ the monitor cannot SP.  
In subsequent layers, monitors on any qubit in the $T$ gate's CCM have some probability to stabilizer-purify.
To simplify our estimate, we replace this CCM by a strip of width $w$ around qubit $i$; setting $w=n$ or $w=1$ shall give upper and lower bounds on $\mathbb{P} (d \geq d^* > 1)$, respectively.

We now focus on the $k^{\textrm{th}}$ layer of the TCB and take $\tilde{s}_1$ and $\tilde{Z}_i$ to be the logicals $s_1$ and $Z_i$ time-evolved to this layer by Cliffords and measurements, and $\mathcal{G}=\langle g_2,\ldots, g_n\rangle $ the code's stabilizer group at this layer. 
(That is, unlike in Sec.~\ref{sec:SPTgate}, we now forward-evolve operators to the monitor, instead of backward-evolving the monitor to the $T$ gate.) 
Using now these operators in Eq.~\eqref{eq:Stab_pur_possibs}, computing the SP probability involves assessing the probability that a monitor $Z_j$ belongs to one of the SP-favorable cosets. 

To assess this, we must consider the number $\gamma_j$ of generators needed to express $Z_j$. This is where entanglement properties enter. We summarize the result, based on Ref.~\onlinecite{fattal2004entanglement}, in Theorem~\ref{thm:gamma}, which we prove in App.~\ref{app:bulkMons}.

\begin{theorem}
    \label{thm:gamma}
    For a pure stabilizer state, one can choose stabilizer generators such that a single-qubit Pauli operator $M_j$ on qubit $j$ is expressible as
    \begin{align}
        M_j =  \prod_{i=1}^{\gamma_j} g_i^{\alpha_i} \overline{g}_i^{\beta_i},%
        \label{eq:stabDecomp}
    \end{align}
    up to a $\pm1$ or $\pm i$ prefactor, where $g_i$ are stabilizer generators, $ \overline{g}_i$ are corresponding destabilizers,\footnote{We define a destabilizer $\bar{g}_i$ of generator $g_i$ of stabilizer group $\mathcal{S}$ to be a Pauli operator that anti-commutes with $g_i$ and commutes with all other generators $g_{j\neq i}$ of $\mathcal{S}$~\cite{aaronson2004improved}.} $\alpha_i, \beta_i = 0, 1$, and the number $\gamma_j$ of generators needed satisfies
    \begin{align}
        \gamma_j = 2 S_\mathrm{vN}(j) + \mathcal{O}(1), \label{eq:gamma}
    \end{align}
    where $S_\mathrm{vN}(j)$ is the von Neumann entanglement entropy of the subsystem with qubits $1,\dots,j-1$.     
\end{theorem}
Using Theorem~\ref{thm:gamma} (setting $M_j= Z_j$), we find that monitor measurement operator $Z_j$ is expressible in terms of $\gamma_j$ generators and their corresponding $\gamma_j$ destabilizers. 
Different monitors may have different $\gamma_j$.
Here, we focus on a regime with a volume- or area-law $S_\mathrm{vN}(j)$; thus, each $\gamma_j$ has the same scaling with the system size $n$.
We shall be interested in this scaling thus we take $\gamma_j= \gamma$ for all $j = j_1 , \dots , j_{w}$ for simplicity (taking $\gamma=\max_j{\gamma_j}$ or $\gamma=\min_j{\gamma_j}$ allows for probability bounds).

We assume the monitor is a uniformly random choice from these $2^{2\gamma_j}-1$ operators; this becomes increasingly true upon increasing $k$. 
Counting the SP-favorable cases conditioned on previous measurements not stabilizer-purifying ($\textrm{NSP}$)---a monitor cannot stabilizer-purify if a previous one already has---thus yields (see also App.~\ref{app:bulkMons})
\begin{align}
    \mathbb{P}(Z_j \textrm{ SP} |\,\textrm{prev. NSP}) = \frac{3}{2}\frac{2^{\gamma_j}}{4^{\gamma_j} -1}.
\end{align}%
Continuing for all the $pw$ potentially purifying monitors in the $k^{\textrm{th}}$ layer, approximating $\gamma_j = \gamma$, and thus denoting $\mathbb{P}(Z_j \textrm{ SP} |\,\textrm{prev. NSP}) \equiv f$,  we find
(see also App.~\ref{app:bulkMons}) 
\begin{align}
    \mathcal{P} \equiv \mathbb{P}(d^* = k |\,k'<k \, \textrm{NSP}) \approx 1- \left( 1 - f  \right)^{pw}, \label{eq:fDef}
\end{align}%
conditioned on previous layers not stabilizer-purfiying. (The result is $k$-independent since we took constant $w$.)%

From $\mathbb{P} (d^* = 1) = p_1$, and in terms of the exact value of $\mathcal{P}$ we have $\mathbb{P}(d^* = k) = (1-p_1)\mathcal{P}(1-\mathcal{P})^{k-2}$  for $k\geq2$. Hence, from
\begin{align}
    \mathbb{P} (d\geq d^* >1) &= \sum_{k=2}^d \mathbb{P}(d^* = k)
\end{align}%
we have the exact relation
\begin{align}
    \mathbb{P}(\textrm{SP}) &= p_1 +(1-p_1)\left[ 1- (1-\mathcal{P})^{d-1} \right]\\
    &= 1 - (1-p_1)(1-\mathcal{P})^{d-1}, \label{eq:SPprob}
\end{align}
using which we can estimate $\mathbb{P}(\textrm{SP})$ via Eq.~\eqref{eq:fDef}. 
\subsubsection{Stabilizer-purification time and entanglement}
As $d$ becomes large, Eq.~\eqref{eq:SPprob} implies $1-\mathbb{P}(\textrm{SP})\propto e^{-\Gamma d}$, with decay rate $\Gamma=-\ln(1-\mathcal{P})$. This allows us to define the stabilizer-purification time  $\tau_{\textrm{SP}}=\Gamma^{-1}$. 
By Eq.~\eqref{eq:fDef}, 
\begin{align}\label{eq:tau_SP}
    \tau_{\textrm{SP}}\approx \frac{-1}{p w\ln(1-f)},\quad f=\frac{3}{2}\frac{2^{\gamma}}{4^{\gamma} -1},
\end{align}
where we reminded of the definition of $f$.

We can now use $\tau_{\textrm{SP}}$ to assess what area- and volume-law EE implies about there being $T$ purifiers and thus SP. We shall use Theorem~\ref{thm:gamma} to infer the scaling of $\gamma$ and thus of $\tau_{\textrm{SP}}$ with the system size $n$. 

In a volume-law phase, $S_\mathrm{vN}(\rho_B) \propto |B|$ for any subsystem $B$ and the subsystems relevant to Theorem~\ref{thm:gamma} have $|B| \propto n$. Thus, $\gamma \propto n$, and 
\begin{align}\label{eq:t_SP_vol}
    \tau_\text{SP} \sim \frac{\exp (n)}{p w} .
\end{align}
Therefore, using $1\leq w \leq n$, we conclude $\tau_\text{SP}=\exp (n)$ in the volume-law phase.
Thus, as each TCB has depth $d=\text{poly}(n)$ [otherwise $D\neq\text{poly}(n)$], finding $T$ purifiers is unlikely, and, over the full circuit, magic accumulates.

Conversely, in an area-law phase, $S_\mathrm{vN}(\rho_B) = \mathcal{O}(1)$ for any contiguous subsystem $B$; thus, $\gamma = \mathcal{O}(1)$.
Hence, by $w_0 \leq w \leq n$, we find that $\tau_{\textrm{SP}} = \mathcal{O}(1/w_0)$ is at most a constant: we have $\tau_\textrm{SP} = \mathcal{O}(1)$ and $\tau_\textrm{SP} = \mathcal{O}(n^{-1})$, for $w_0 = \mathcal{O}(1)$ and $w_0 = \mathcal{O}(n)$, respectively.
Now, as each TCB has depth $d=\text{poly}(n)\gg \mathcal{O}(1)$, it is almost sure that each TCB is a $T$ purifier, thus magic cannot accumulate. 

These scalings of $\tau_{\textrm{SP}}$ with $n$ in area- and volume-law phases match that of the (entropy) purification time in Ref.~\onlinecite{gullans2020purification} for pure and mixed phases, respectively.

\subsubsection{Numerical test}
We test our predictions for the SP time via a numerical experiment. 
We take a circuit on $n$ qubits (initialized in a computational basis state) that consists of (i) a depth $n^2$ brickwork of random $2$-qubit Clifford gates and monitors and then (ii) a TCB
of varying depth $d$ with random $2$-qubit Cliffords and monitors. Monitors are sampled independently with probability $p$. 
The circuit block before the $T$ gate generates a random stabilizer state ($\ket{\psi_\text{in}}$ in our above construction, cf. Fig.~\ref{fig:TCB}) with volume- or area-law entanglement depending on $p$, while the TCB probes whether SP occurs. 
Specifically, we are interested in numerically estimating $\Gamma$ and thus $\tau_{\textrm{SP}}=\Gamma^{-1}$.

Our simulations agree with the expectations: 
in a volume-law phase (e.g., for $p=0.1<p_c^{\textrm{EE}} \approx 0.16$ as in Fig.~\ref{fig:SP_Prob_volume_law}) we find that $1-\mathbb{P}(\mathrm{SP})$ decays exponentially with $d$ and that $\Gamma \sim \exp(-n)$ (Fig.~\ref{fig:SP_Prob_volume_law} inset), both consistent with Eq.~\eqref{eq:t_SP_vol}.
Conversely, in an area-law phase (e.g., for $p=0.2, 0.4 > p_c^{\textrm{EE}} $ as in Fig.~\ref{fig:SP_Prob_area_law}) we find that $\mathbb{P}(\mathrm{SP})$ saturates to $1$ in a depth $d$ independent of $n$ and decreasing with $p$, as predicted by Eq.~\eqref{eq:tau_SP} with $\gamma = \mathcal{O}(1)$. 
\begin{figure}[t]
    \centering
        \includegraphics[width=8.6cm]{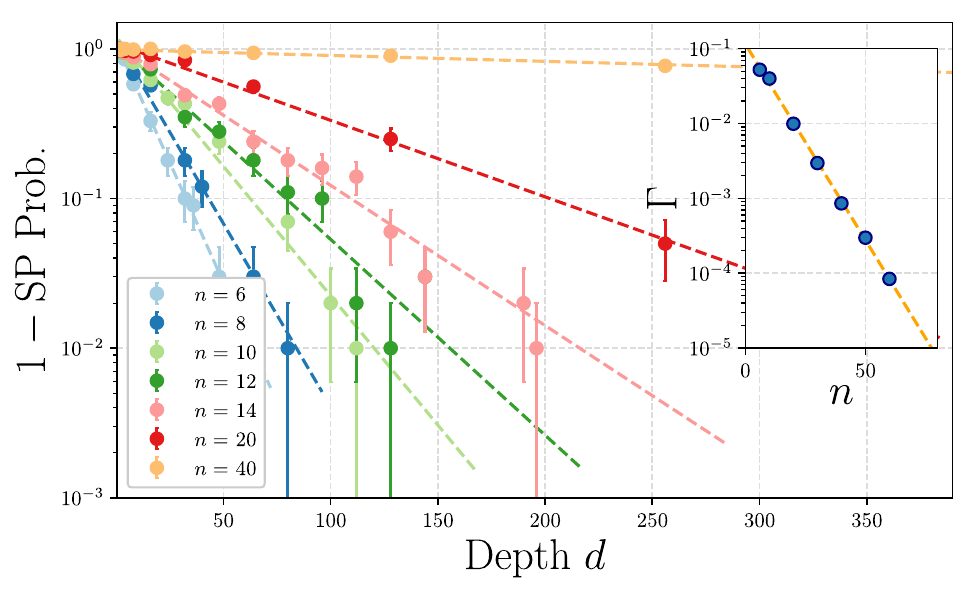}
    \caption{
                Stabilizer-purification probability $\mathbb{P}(\textrm{SP})$ as a function of TCB depth $d$ for $p=0.1$ (volume-law phase), with SE as error bars. 
                The dashed lines are fits of $c e^{-\Gamma(n) d}$ to $1-\mathbb{P}(\textrm{SP})$, which agree with the analytical considerations. 
                \textit{Inset}: Dependence of the decay rate $\Gamma (n)$ on system size $n$. 
                The fit $\exp(-n)$ (orange dashed line) agrees well with our theory.
            }
    \label{fig:SP_Prob_volume_law}
\end{figure}
\begin{figure}[t]
    \centering
        \includegraphics[width=8.6cm]{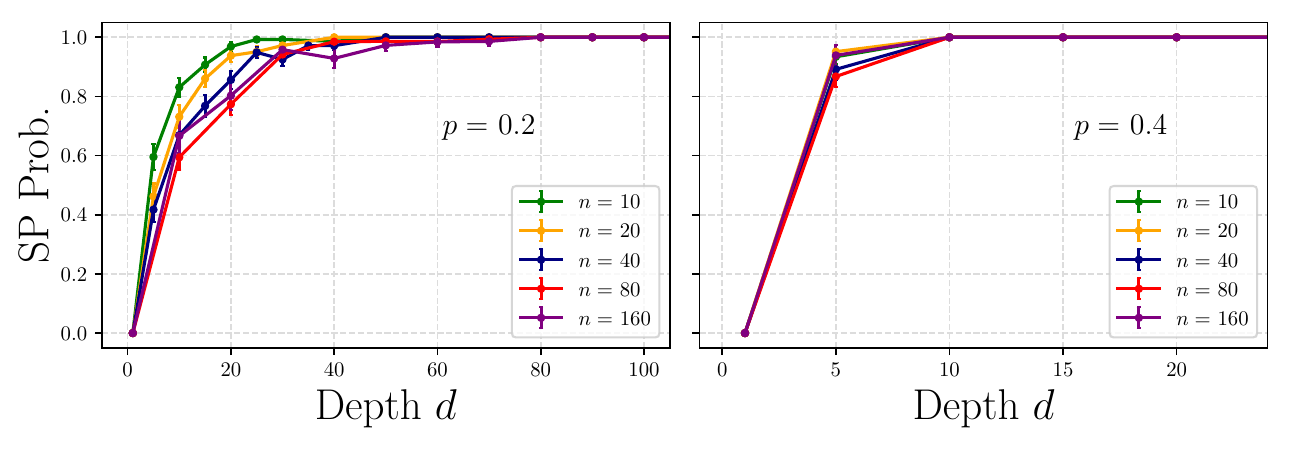}
    \caption{
                Stabilizer-purification probability $\mathbb{P}(\textrm{SP})$ as a function of TCB depth $d$ for $p=0.2$ (left) and $p=0.4$ (right), so in the area-law phase. 
                (Error bars: SE.)
                In both panels, $\mathbb{P}(\textrm{SP})$ saturates to 1 in a depth $d$ independent of the system size $n$ and decreasing with larger $p$, which matches our expectations.
            }
    \label{fig:SP_Prob_area_law}
\end{figure}

\subsection{SP probability implications for \texorpdfstring{$\mathrm{\textsf{CPX}_{\mathrm{\bf{PBC}}}}$}{} \label{sec:SPtime_implications}}
In this section, having built some intuition for $\tau_\text{SP}$, we focus on its implications for magic transitions for fixed $q$ or $qD$. 
We revisit the numerical results suggesting a transition for fixed $qD$ from Sec.~\ref{sec:CPXtrans} and use SP to explain them.
Then, we present numerical results suggesting the absence of a simulability transition for fixed $q$ below the percolation threshold, and show that this is expected from our analytical considerations. 
These also suggest that SP is the leading mechanism for MF and the existence of a magic transition.

\subsubsection{Fixed \texorpdfstring{$q$D}{} \label{sec:SP_AC}}
Let us first consider fixed $qD = \mathcal{O}(1)$. 
For concreteness, we take $D = c n^a$ to leading order in $n$, with $a > 0$ and $c$ a constant.
Since $q$ is the spacetime density of $T$ gates, each $T$ gate is typically the sole occupant of an $\mathcal{O}(1/q) = \mathcal{O}(n^{a})$ spacetime volume.
Let us define the expected radius of one of these regions to be $r_\text{exp} = \mathcal{O}(n^{a/2})$. We also define the expected number of layers between $T$ gates as $d_\text{exp} = \mathcal{O}(n^{a-1})$.
Since these capture the separation between causally connected $T$ gates for $w_0 = \mathcal{O}(1)$ and $w_0 = \mathcal{O}(n)$, respectively, we expect the simplified model from Sec.~\ref{sec:SPTime_purifier} to apply with $d=d_\text{exp}$ for $w_0 = \mathcal{O}(n)$, and with $d = r_\mathrm{exp}$ for $w_0 = \mathcal{O}(1)$.
We use this to show that, for uncorrelated monitoring and fixed $qD = \mathcal{O}(1)$, we expect a magic transition, evidenced by $\mathrm{\textsf{CPX}_{\mathrm{PBC}}}$, coinciding with the entanglement transition.

In the volume-law phase, $\tau_{\textrm{SP}} = \exp (n)\gg D$ [Eq.~\eqref{eq:t_SP_vol}].
Therefore, since $d_\text{exp}/\tau_\mathrm{SP} = \exp(-n)$, magic from many $T$ gates spreads and builds up in the system before any one of them could have its magic removed by monitor measurements; PBC features a MSR block of size $\propto qDn \propto n$ to simulate (cf. proof of Theorem~\ref{thm:PBC}), thus a hard phase.

Conversely, the area-law phase has $\tau_{\textrm{SP}} = \mathcal{O}(1/w_0)$.
We can focus on cones with constant-width apex since stabilizers are generically local.\footnote{Local stabilizers are typical since we start from a product state and the frequent monitors keep the state in the area law phase. Even if an atypical nonlocal generator occurs, leading to a cone with an $\mathcal{O}(n)$ apex, it may SP even faster in time $\tau_\mathrm{SP} = \mathcal{O}(1/n)$ if $w_0 = \mathcal{O}(n)$. In this case, $T$ gates also SP individually since $d_\mathrm{exp}/ \tau_\mathrm{SP} = \mathcal{O}(n^a)$, for $n \gg 1$ and $a>0$.}
Hence, we apply the model from Sec.~\ref{sec:SPTime_purifier} with $d = r_\mathrm{exp}$ and $\tau_\textrm{SP}=\mathcal{O}(1)$, thus $d/\tau_\textrm{SP} = \mathcal{O}(n^{a/2})$.
Hence, for $n\gg1$ and $a>0$, each $T$ gate is individually stabilizer-purified independently of all other $T$ gates. %
Thus, the magic from the $T$ gates in the bulk of the circuit is rapidly purified; in PBC these contribute $\mathcal{O}(1)$-weight MSR fragments. 
The magic in the final state comes only from the $\sim qn\tau_{\textrm{SP}} =(qD)n\tau_{\textrm{SP}}/D= \mathcal{O}(n^{1-a})$ $T$ gates within $\tau_{\text{SP}}$ from the end of the circuit.
These $T$ gates are typically spacelike separated from each other (for constant $w_0$), thus in PBC they yield $\mathcal{O}(1)$-weight MSR fragments (see Sec.~\ref{sec:CpxProxy}).
Thus, MF occurs, leading to an easy phase.

Hence, we expect EE and magic transitions to coincide in this regime. 
The mechanism presented here explains the magic transition discussed in Sec.~\ref{sec:CPXtrans}: we have argued SP drives MF, which is related to an easy phase in terms of $\mathrm{\textsf{CPX}_{\mathrm{PBC}}}$.
We identify $\tau_\text{SP}$ as the emergent time scale generalizing a correlation length for the transition out of the easy phase, where $\tau_\text{SP}$ diverges upon approaching the transition as $p \to p_c^+$.

\subsubsection{Fixed \texorpdfstring{$q$}{} \label{sec:qCPX}}
As we now explain, for fixed $q=\mathcal{O}(1)$ and uncorrelated monitoring, we expect no magic transition below the percolation transition $p_c^\text{TN}$
(for $p>p_c^{\text{TN}}$ we find an easy phase; cf. Sec.~\ref{sec:perc} and App.~\ref{app:SpacetimePart}).
In each circuit layer, $qn = \mathcal{O}(n)$ $T$ gates occur on average, and these increase the number of logical qubits encoded in the corresponding effective stabilizer code up to $\mathcal{O}(n)$.
We focus on $p<p_c^{\text{TN}}$ in the area-law phase with $q$ suitably small such that the apexes of the CCMs of the $T$ gates do not overlap (a hard phase in this regime would suggest one for larger $q$, or smaller $p$, or both). 
Since $r_\mathrm{exp} = \mathcal{O}(1/\sqrt{q}) = \mathcal{O}(1)$ here, $r_\mathrm{exp} / \tau_\mathrm{SP} = \mathcal{O}(1)$ does not increase with $n$, thus, a finite fraction of $T$ gates are not stabilizer-purified and can compound their magic in the MSR. 
Hence, we generically expect $\mathcal{O}(n)$ logical qubits to build up and persist throughout the evolution, and a MSR block of size $\sim qDn=\text{poly}(n)$ on which to simulate measurements, leading to a hard PBC phase irrespective of EE.\footnote{This argument cannot rule out an easy phase for $\sqrt{1/q}/\tau_\mathrm{SP} \gg 1$ where CCMs of nearby $T$ gates are unlikely to overlap before SP.
However, testing this, via the scaling with $n$, is numerically challenging since $\tau_\mathrm{SP}$ can be appreciable and thus avoiding finite-size effects from $q \ll 1/\tau^2_\mathrm{SP}$ requires $n$ beyond our reach.}
\begin{figure}[t]
    \centering
        \includegraphics[width=8.6cm]{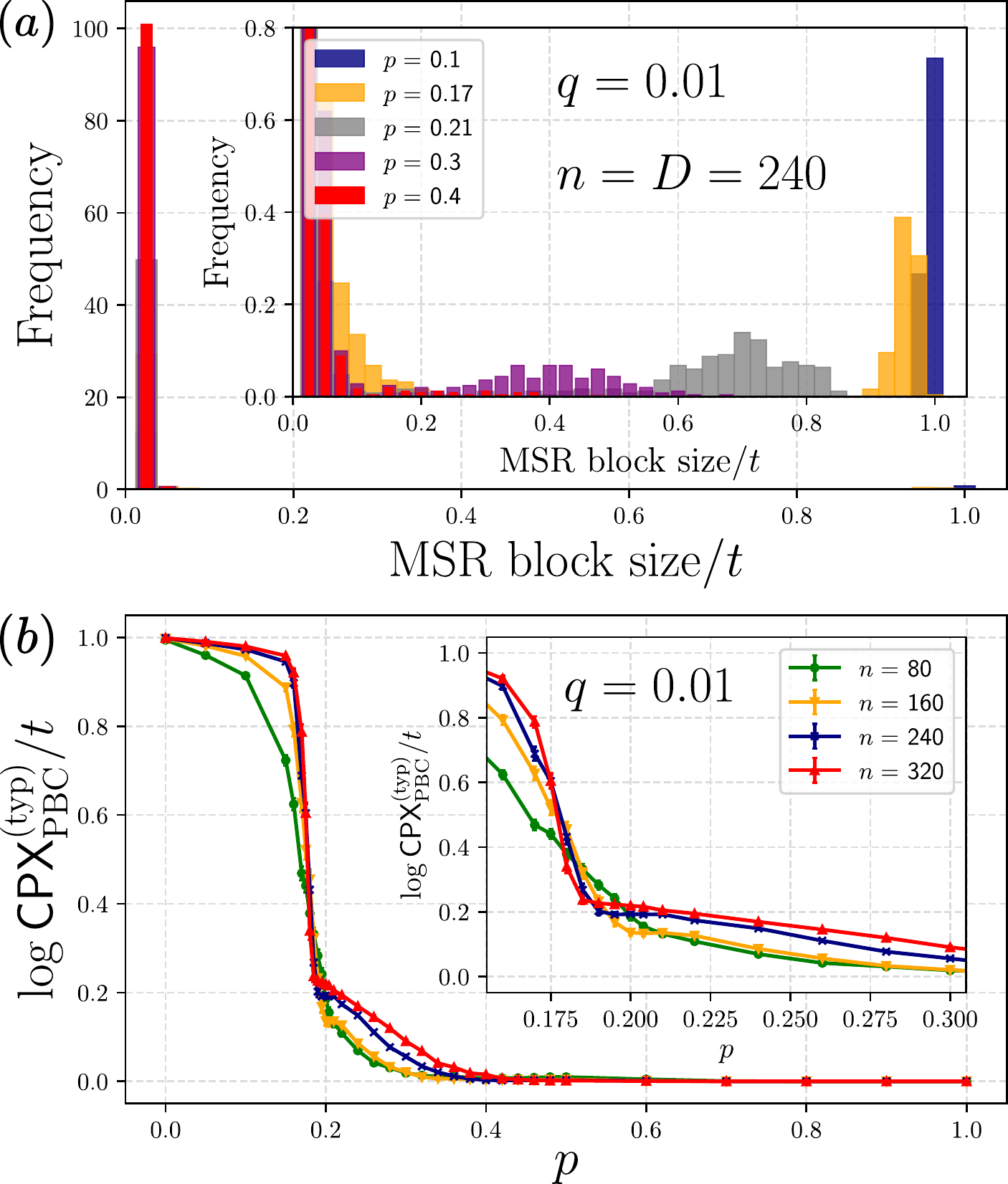}
    \caption{
                Absence of a magic transition for fixed $q$ below the percolation threshold $p_\mathrm{c}^\textrm{TN} = 0.48$. 
                $(a)$: Both in the volume- and area-law phase, the distribution of magic state register block sizes suggests a typical size of at least $\mathcal{O}(n)$, leading to a hard phase. The inset more closely shows the histogram of block sizes averaged over 300 realizations.
                $(b)$: The order parameter $\log \mathrm{\textsf{CPX}}^{(\mathrm{typ})}_{\mathrm{PBC}}/t$ versus measurement probability $p$ (with SE as error bars), where the runtime proxy for simulating a circuit by the PBC method is $\mathrm{\textsf{CPX}}_{\mathrm{PBC}}$, and $t$ is the number of $T$ gates. 
                The order parameter remains finite as the entanglement transition is encountered, and it increases with $n$ even in the area law below the percolation threshold.
                The inset shows data closer to the entanglement transition, suggesting the absence of a magic transition.
            }
    \label{fig:q_CPX}
\end{figure}
We next show numerical evidence for a hard phase for $p < p_\mathrm{c}^\textrm{TN}$, focusing on the circuits in Sec.~\ref{sec:CPXtrans}, but now with $qD\sim D$.
As shown in Fig.~\ref{fig:q_CPX}$(a)$ and its inset, although the distribution of MSR block sizes shifts towards lower values as  $p$ increases, it remains peaked at a block size proportional to the total MSR size $t = \mathrm{poly}(n)$; this suggests MF does not occur. 
Looking at the simulability order parameter, we observe it crosses over from a hard and volume-law entangled phase to a hard and area-law entangled phase at $p\leq p_\mathrm{c}^\mathrm{EE}\approx0.17$ (the Haar-random value~\cite{zabalo2020EEpc}), cf. Fig.~\ref{fig:q_CPX}$(b)$ and its inset. 
Even though the order parameter significantly decreases in the area-law phase, it remains finite upon increasing $n$; the hardness of simulation persists until $p_\mathrm{c}^\mathrm{TN}=0.48$. 
These results suggest that 
(i) the absence of SP leads to no MF and no magic transition for $p<p_\mathrm{c}^\mathrm{TN}$, 
and (ii) the magic and entanglement transitions are distinct.
\subsection{SP implications for direct stabilizer simulations \label{sec:SPdirect}}
While thus far we mostly linked SP to PBC simulations, here
we explain that SP also implies easy stabilizer simulations for the original circuit. Concretely, we show that if at most $\mathcal{O}(\log n)$ $T$ gates occur per layer, and the magic from each $T$ gate is stabilizer-purified in $\mathcal{O}(1)$ time (or vice versa), then stabilizer simulation is easy.

Consider a circuit as that from Fig.~\ref{fig:circuit}$(a)$. 
Under our assumptions, at any point in the time-evolution, the $T$ gates whose magic has not yet been stabilizer-purified encode $\mathcal{O}(\log n)$ logical qubits. 
This implies that, at any point, the system is in a superposition of $\exp[\mathcal{O}(\log n)]=\text{poly}(n)$ stabilizer states; keeping track of these via stabilizer methods over depth $D=\text{poly}(n)$ takes $\text{poly}(n)$ classical runtime and memory. 
Hence, by simulating the time-evolution via stabilizers, we efficiently find the exact state $\ket{\psi_\text{out}}$ before the final computational basis measurements.  
As $\ket{\psi_\text{out}}$ is a superposition of $\text{poly}(n)$ stabilizer states, weak or strong simulation can be done efficiently.
\section{\texorpdfstring{$\bm{T}$}{}-correlated Monitoring}
\label{sec:MonGame}
We next discuss a model where correlations between $T$ gates and monitors facilitate SP, thus enabling a magic transition within a volume-law phase. 
Thus, the magic transition is now a simulability transition, occurring without a phase transition in EE. 

We shall use the $T$-correlated monitoring model (see Sec.~\ref{sec:InfMon})
with the circuit depicted in Fig.~\ref{fig:circuit}$(a)$. 
In this model, we consider the conditional probability $p_+ = \mathbb{P}(Z_j|T_j)$ of applying a $Z_j$ monitor given there is a $T$ gate $T_j$ on qubit $j$ directly preceding it, and the conditional  $p_- = \mathbb{P}(Z_j| \textrm{no } T_j)$ for there being no directly preceding $T_j$.
The probability of a $T$ gate is still $\mathbb{P}(T_j)=q$ independently for each qubit $j$, and we still have 
\begin{equation}\label{eq:pm_prob}
\mathbb{P}(Z_j)=p_+ q + p_- (1-q) = p,
\end{equation}
independently for each qubit, for the total probability of applying a monitor $Z_j$. 
However, we can now have $p_+\neq p_-$: monitors can be correlated with $T$ gates. [For $p_+=p_-$, Eq.~\eqref{eq:pm_prob} implies $p_+=p_-=p$ and we recover the uncorrelated monitoring model from Sec.~\ref{sec:CPXtrans}.]

In what follows we parameterize  $p_+ = p_- + \alpha$, with $\alpha$ independent of $q$, and take $\alpha\geq 0$. (In terms of a monitoring observer, cf. Sec.~\ref{sec:InfMon}, this expresses the aim to stabilizer-purify the $T$ gates; $\alpha<0$ would mean monitoring while trying to avoid SP.)
From Eq.~\eqref{eq:pm_prob} we have $p_- = p - \alpha q$, thus we recover $p_- = p$ as $q \to 0$, i.e., the correct limit without $T$ gates ($p_+$ plays no role for $q=0$).
From $0\leq p_\pm \leq 1$ we find $0\leq \alpha\leq \text{min}(\frac{1-p}{1-q},\frac{p}{q} )$.
For concreteness, we focus on $p\geq q$; in this case the upper bound is $\alpha_\text{max}=\frac{1-p}{1-q}$.

We start with the limits $\alpha=0$ and $\alpha=\alpha_\text{max}$.  
For $\alpha = 0$, we recover the uncorrelated monitoring model from Sec.~\ref{sec:CPXtrans}. 
If $p<p_c^\text{EE}$, i.e., the system is in a volume-law phase, then $\alpha=0$ yields a hard to simulate phase for any nonzero $q$, cf. Sec.~\ref{sec:SPtime_implications}. 
For $\alpha = \alpha_\text{max}$, we find $p_+ = 1$. 
This ``perfect monitoring" limit is easy to simulate since the magic from each $T$ gate is immediately stabilizer-purified in each monitoring round.
This holds for any $q$ [provided $q\leq p$, as required for $\alpha = (1-p)/(1-q)$ to be consistent], including in the volume-law phase. 
\begin{figure}[t]
    \centering
        \includegraphics[width=8.6cm]{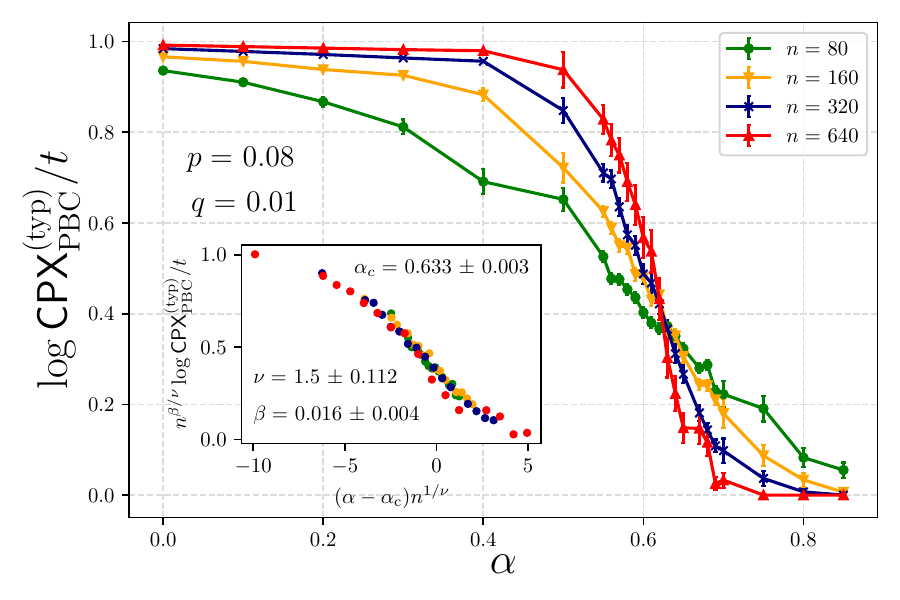}
    \caption{
                Dynamical magic phase transition for $T$-correlated monitoring. 
                The probability for applying a $T$ gate is $q=0.01$, while for monitors, it is $p=0.08$, and $\alpha$ parametrizes their correlation. 
                As $p<p_\mathrm{c}^\mathrm{EE} \approx 0.16$, the system is in a volume-law phase.
                We plot the order parameter $\log \mathrm{\textsf{CPX}}^{(\mathrm{typ})}_{\mathrm{PBC}}/t$, with SE as error bars. 
                The order parameter drops from a finite value for $\alpha<\alpha_\mathrm{c}$ to zero (as $n\to\infty$) for $\alpha>\alpha_\mathrm{c}$.  
                The inset shows a finite-size scaling collapse revealing the critical $\alpha_\mathrm{c}=0.633 \pm 0.003$, well below the perfect monitoring value $\alpha_\text{max}=(1-p)/(1-q)\approx 0.93$ for $p=0.08$, $q=0.01$.
            }
    \label{fig:alpha_CPX}
\end{figure}
As we reduce $\alpha$ from  $\alpha_\text{max}$, we expect an easy to simulate phase to persists, at least for sufficiently small $q$, even if $q$ is $n$-independent. 
To test this, we perform a numerical experiment similar to that in Sec.~\ref{sec:CPXtrans}, now focusing on the volume-law phase. 
Our numerical results, illustrated in Fig.~\ref{fig:alpha_CPX} for $p=0.08$ and $q=0.01$, suggest that both easy and hard phases are stable, and there is a magic transition, which is now a simulability transition, separating them, despite the system being in the volume-law phase. 
The phase transition is continuous, as corroborated by the scaling collapse, cf. Fig.~\ref{fig:alpha_CPX} inset. 
We find a critical value of $\alpha = \alpha_\mathrm{c} = 0.633 \pm 0.003$, well below $\alpha_\text{max}\approx 0.93$ for $p=0.08$ and $q=0.01$. 
\section{Discussion and Outlook \label{sec:Discussion}}
We have studied how the dynamics of magic in random monitored Clifford+$T$ circuits impacts classical simulability and, in particular, how monitoring measurements may lead to the spreading of magic becoming arrested (and indeed magic being removed) by a process we dubbed stabilizer-purification (SP).
We used PBC
to quantify the role of magic in classical simulabilty, and identified magic fragmentation (MF),
linked to SP, as the key phenomenon behind the transition from hard to easy PBC phases. 
Concretely, we showed that SP implies MF in a simplified model in Theorem~\ref{thm:PBC}, argued how this extends to our circuit model, and provided numerical evidence supporting these claims; 
this leaves the establishment of a more formal link between SP and MF in a broader class of circuits as an interesting problem for the future.
The dynamics of magic, and the concepts of SP and MF, open new avenues for investigating phase transitions in the complexity of simulating quantum circuits beyond the paradigm of entanglement. 
Here, we showed that they can lead to a simulability transition within a volume-law entangled phase (as exemplified by $T$-correlated monitors), but also found scenarios where the magic transition occurs within an area-law phase (fixed $q$) or it coincides with the entanglement transition (fixed $qD$).
While approximate simulation is always possible in an area-law phase owing to MPS methods~\cite{vidal2003efficient,vidal2004efficient}, exact efficient simulation requires the Hartley entropy to obey an area law~\cite{schuch2008entropy}, which occurs only above the critical probability we have called $p_c^{\text{TN}}$ (cf. App.~\ref{app:clusters}).
By varying the number of $T$ gates in our model, one could interpolate between Clifford and universal circuits, potentially approaching Haar-random circuits
(since circuits with more $T$ gates can form higher unitary $k$-designs~\cite{qHomeopathy}). 
Taking the perspective of a monitoring observer introducing the correlations between $T$ gates and measurements, it would be intriguing to study how much steering~\cite{morral2023entanglement} and learning~\cite{ChargeSharpening2022} capacity the observer has. 

The metric $\mathrm{\textsf{CPX}}_{\mathrm{PBC}}$, behind our simulability order parameter, is also a magic metric. 
It is not a magic monotone however (since it can increase under Clifford gates); instead it quantifies the amount of magic that has \textit{spread} in the circuit.
This spread is essential for quantum advantage.
To illustrate this with an extreme example, consider the state $\ket{\psi}=(T\ket{+})^{\otimes n}$. 
Some magic metrics would indicate there is an extensive amount of magic in $\ket{\psi}$, yet sampling from its output distribution---requiring only $n$ independent coin tosses---is clearly a classically easy problem. 
The problem is easy because in $\ket{\psi}$ the magic from each $T$ gate is localized to its respective qubit.
The metric $\mathrm{\textsf{CPX}}_{\mathrm{PBC}}$, by capturing the (de)localization of magic, identifies such cases as classically efficiently simulable. 
By detecting the magic that can yield quantum advantage, we may view $ \mathrm{\textsf{CPX}}_{\mathrm{PBC}}$ as a metric for ``operational magic".

This metric also avoids the post-selection problem of measurement-induced phase transitions~\cite{IppKhemaniPost,LiPostPRL,naturePOST2023,garratt2023probing, mcginley2023postselectionfree} since $\mathrm{\textsf{CPX}}_{\mathrm{PBC}}$ does not depend on the measurement outcomes: 
If the outcome of a mid-circuit monitor is changed, this can at most change the signs, but not the structure, of measurement operators in PBC, hence leaving $\mathrm{\textsf{CPX}}_{\mathrm{PBC}}$ unaffected.
Also, $\mathrm{\textsf{CPX}}_{\mathrm{PBC}}$ of a circuit instance can be computed classically in $\mathrm{poly}(n,t)$ time; hence, the order parameter could be used as a diagnostic of magic transitions,  and for $qD=\mathcal{O}(1)$ also of entanglement transitions, accessible to classical simulators. 
The stabilizer-impure phase can be interpreted as a non-stabilizer state encoded in a dynamically generated stabilizer code~\cite{choi2020qecc,gullans2020purification,Hastings2021dynamically}. 
As we saw in Sec.~\ref{sec:SPTgate}, already a single $T$ gate can yield such an encoding, provided it increases the stabilizer rank.
Clifford gates and measurements, which do not stabilizer-purify, dynamically modify this logical subspace by updating the stabilizer generators. 
Monitors that stabilizer-purify act as logical errors
decreasing the logical subspace's dimension; thus, they compete with the encoding $T$ gates. 
This picture proved fruitful for our discussion, and it would be interesting to see whether a quantum error correcting code perspective would allow statistical mechanical mappings~\cite{DennisKitaev, chubb2021statistical,SMofQECC,VBB2023,BVB2023}, that could give a complementary understanding of the dynamics of magic, SP, and magic transitions. 
Such a mapping might allow one to contrast the universality class of the magic transition to that of entanglement.
These may prove to be the same for uncorrelated monitoring with $qD=\mathcal{O}(1)$, since the values of $p_\text{c}\approx 0.159$ and $\nu \approx 1.23$ we found are consistent with those for entanglement transitions in Clifford circuits~\cite{zabalo2020EEpc}.

We may also view $T$ gates as coherent errors on an encoded stabilizer state~\cite{greenbaum2017modeling,bravyi2018correcting,HuangDohertyFlammia19,iverson2020coherence,VBB2023,BVB2023,magic8pt2}; from this viewpoint magic is a coherent, pure-state analog of the entropy that would come from Pauli channels. 
This suggests interesting directions in the entropy-purification settings~\cite{gullans2020purification}.
In particular, in our setups magic is injected throughout the time-evolution; this does not directly correspond to the original dynamical purification setup~\cite{gullans2020purification}, where all entropy is injected at the start (i.e., the input is a maximally mixed state). 
Although the purification and entanglement phase transitions were found to coincide in (1+1)D and (2+1)D for Clifford circuits~\cite{gullans2020purification, Purif2D, MIPTs_d_plus_one}, where they can be mapped to the same statistical mechanics model~\cite{CFTmapping, jian2020mipt}, these two transitions might generically differ.
Building on our settings, it would be appealing to attempt separating the purification and EE transitions by having a pure input state and dynamically mixing the state (i.e., decreasing the state's purity mid-circuit). 
Entanglement in the MSR can also display signatures of the dynamics of magic in the original circuit.
Consider the entanglement in the MSR after all the monitor measurements.
If MF occurs, then this final MSR state has non-overlapping $\mathcal{O}(1)$-weight measurements; when these are local (as they are for magic states ordered lexicographically following how $T$ gates occur, and by SP such that the original circuit regularly has layers with stabilizer output), the MSR obeys the area-law:  
the MSR is in an area-law stabilizer state, possibly tensored with a local $\mathcal{O}(1)$-sized unmeasured MSR block.
Conversely, if MF does not occur, the MSR is expected to obey a volume-law since most measurements have support on an extensive number of magic states.
This suggests that magic and entanglement transitions may be unified if one focuses on the entanglement properties in the MSR.
Exploring this is another interesting direction for the future.

\textit{Note added:} Independently of this work, Fux, Tirrito, Dalmonte, and Fazio~\cite{FTDF23} also studied magic and entanglement transitions in a similar setup. Using a different approach, they also find that these transitions occur for different $p$ in general. However their approach suggests separate transitions also for $qD=\mathcal{O}(1)$ with uncorrelated monitoring, whereas in that case we find simultaneous  transitions, cf. Secs.~\ref{sec:CPXtrans}, \ref{sec:SPtime_implications} and App.~\ref{app:add_qD}. 

\begin{acknowledgments}
We thank Jan Behrends, Matthew Fisher, Richard Jozsa, Ioan Manolescu, Sergii Strelchuk, Florian Venn, and Peter Wildemann for useful discussions.
This work was supported by EPSRC grant EP/V062654/1 and two EPSRC PhD studentships.
Our simulations used resources at the Cambridge Service for Data Driven Discovery operated by the University of Cambridge Research Computing Service (\href{www.csd3.cam.ac.uk}{www.csd3.cam.ac.uk}), provided by Dell EMC and Intel using EPSRC Tier-2 funding via grant EP/T022159/1, and STFC DiRAC funding (\href{www.dirac.ac.uk}{www.dirac.ac.uk}).
\end{acknowledgments}

\appendix
\section{Details of the PBC method \label{app:PBC_Details}}

In this Appendix, we provide details for the PBC method we use for simulating our Clifford$+T$ circuits (acting on an $n$-qubit register $\mathcal{R}_n$), which is based on Refs.~\onlinecite{Bravyi_PBC2016, Yoganathan2019QuantumAO}.
Specifically, we explain how measurements can be restricted to act only non-trivially on the magic state register ($\mathcal{R}_t$).

As explained in Sec.~\ref{sec:CpxProxy}, we start with a monitored circuit with random Clifford gates and $t$ applications of the $T$ gate. 
We then replace all $T$ gates with magic state gadgets, each using a magic state ancilla $\ket{A}$ to inject the $T$ gate into the circuit (Fig.~\ref{fig:PBC_fig}).
After doing this, we commute all Clifford gates past all measurements, performing updates $M\mapsto C^\dagger M C$ for measurement operators $M$ and Clifford gates $C$. 
Once the Clifford gates are thus commuted to the end of the circuit, they can be deleted.

Let $M_i$ denote the measurement operators resulting from this process. 
To this set of measurements we also add a series of ``dummy measurements" to the start of the circuit, which are simply $Z$ measurements on all qubits in the computational register.
Owing to the initial state $\ket{0}^{\otimes n}$ of this register, these dummy measurements produce outcomes $+1$ with certainty. Let $\mathcal{S} = \langle Z_1, \dots , Z_n \rangle$ denote the group generated by these $Z$ operators.

The entire list of measurements can be restricted to the magic state register in the following way.
For each non-trivial measurement operator $M_i$, let $M_i = P_i Q_i$, where $P_i$ ($Q_i$) only has support on $\mathcal{R}_n$ ($\mathcal{R}_t$). %
We begin with $M_1$.

First, suppose $P_1$ commutes with all previously performed measurements (which are simply dummy measurements).
Then $P_1$ belongs to either $\pm \mathcal{S}$.
If $Q_1 = \mathds{1}$ the measurement outcome is deterministic and can be computed efficiently classically.
If $Q_1$ is non-trivial, then $Q_1$ has the same measurement statistics as the entire operator $M_1$ (up to a potential change of sign).
Hence $M_1$ can simply be replaced with $\pm Q_1$ without altering the measurement statistics or post-measurement state, with the sign determined by $M_1$ belonging to $\mathcal{S}$ or $-\mathcal{S}$.

Second, suppose $P_1$ anti-commutes with some $Z_k\in\mathcal{S}$.
In this case, it is simple to see that outcomes $M_1 = \pm 1$ each occur with probability $1/2$; thus we can simulate its measurement with an unbiased coin with outcome $\lambda_1 \in \lbrace +1, -1\rbrace$.
Instead of performing the measurement of $M_1$, we can enact the Clifford gate $V_{\lambda_1} = \exp (\lambda_1 \frac{\pi}{4} M_1 Z_k) = \frac{1}{\sqrt{2}}(1 + \lambda_1 M_1 Z_k)$.
This maps the initial state of $\mathcal{R}_n\otimes \mathcal{R}_t$ to the post-measurement state associated with $M_1 = \lambda_1$, since $Z_k$ stabilizes the initial state.
After having performed this replacement, we commute $V_{\lambda_1}$ past all remaining measurements in the circuit (thereby updating them to other Pauli measurements) and then delete it. 

We proceed similarly for all measurements. 
For each (updated) $M_i$ we first check if this measurement operator is independent of any previously performed measurements (including the dummy measurements). 
If $M_i$ is equal (up to a sign) to a product of previous measurements, we need not perform the measurement of $M_i$ explicitly. Instead its measurement outcome is deterministic and can be computed efficiently classically.
We then check if it anti-commutes with any previously performed measurements. 
If not, it can be restricted to $\mathcal{R}_t$ for the same reason as above.
If so, it can be replaced by some $V_{\lambda_i}$ with $\lambda_i \in \lbrace +1,-1\rbrace$ chosen uniformly at random, as above.
If it anti-commutes with previous measurement $N$ with outcome $\sigma$, we choose $V_{\lambda_i} = \exp(\lambda_i \sigma \frac{\pi}{4} M_i N)$.
$V_{\lambda_i}$ is commuted past all remaining measurements and deleted.

After proceeding in the same way for all $M_i$, we end up with a set of measurements restricted to the register $\mathcal{R}_t$. 
We can now delete the computational register $\mathcal{R}_n$ which no longer features in the circuit.
$\mathcal{R}_t$ is composed of $t$ qubits. 
The measurements resulting from the above procedure all commute since anti-commuting measurements were replaced by Clifford gates.
Therefore we end up with at most $t$ commuting (adaptive) measurements needing to be performed on $\mathcal{R}_t$.

\subsection{Runtime of classically simulating a PBC}

Naively the runtime of the classical simulation of the PBC will be $\mathcal{O}(t^3 \chi_t^2)$ for each measurement in the final PBC being simulated~\cite{Bravyi_PBC2016}, plus the time it takes to calculate the next measurement in the PBC from previous measurement outcomes and check if it is independent from previously performed measurements [which is a poly$(t)$-time task].
While we are ultimately concerned with the exponential part of the simulation runtime, we first note some simplifications that could be made to the simulation, which will come into play for our numerical simulations.

Suppose there are $k$ measurements in the PBC. 
Let $k = k_g  + k_m + k_o$ where we have $k_g$ of the final measurements resulting from original gadget measurements (GMs), $k_m$ resulting from monitoring measurements and $k_o$ from output measurements (OMs).
Simulating a monitored circuit, we assume we know the outcomes of the monitoring measurements already and merely wish to sample from the output distribution of the OMs.
Furthermore, it can be seen that GMs have equal probabilities $1/2$ for outcomes $\pm 1$.
So these two types of measurements from the circuit, if they are retained in the PBC, have output probabilities that do not need to be calculated. %
We only need to perform non-trivial, possibly exponentially-scaling calculations for $k_o$ of the $k$ measurements. 
\section{Details on stabilizer-purified \texorpdfstring{$\bm{T}$}{} gate \label{app:SPTgate}}
In this Appendix, we prove some of the statements used in Sec.~\ref{sec:SPTgate}. 

\subsection{Non-stabilizer superposition}\label{app:non-stab_super}
Here, we show why the non-stabilizer state $T\ket{\psi}$ can be decomposed as a superposition of (at least) two stabilizer states $\ket{\psi_\pm}$ with $\mathcal{S}_\pm = \langle \pm Z_1, g_2, \dots, g_n \rangle$. 

A useful result of Ref.~\onlinecite{bravyi2005universal} is that a single qubit state $\ket{\phi} = a_0 \ket{0} + a_1 \ket{1}$ with $|a_0|=|a_1|=1/\sqrt{2}$ is a stabilizer state iff the phase difference between $a_0$ and $a_1$ is a multiple of $\pi/2$, that is $\arg (a_1/a_0) = m \frac{\pi}{2}$ with $m=-3, \dots, 3$.

Let us consider the overlaps between initial stabilizer state $\ket{\psi}$ with $\mathcal{S} = \langle s_1, \dots, s_n \rangle = \langle s_1, g_2, \dots, g_n \rangle$ and $\ket{\psi_\pm}$. 
Note that we can use for $\ket{\psi}$ also the generators $g_2, \dots, g_n$, which commute with both $s_1$ and $\pm Z_1$. 
Since these are pure states, we have $|\braket{\psi}{\psi_\pm}|^2 = \mathrm{Tr} (\rho \rho_\pm)$, where $\rho_{(\pm)} = \ketbra{\psi_{(\pm)}}{\psi_{(\pm)}}$. 
Writing the density matrices of pure stabilizer states in terms of their generators, we find 
\begin{align}
      \mathrm{Tr} (\rho \rho_\pm) 
                &= \mathrm{Tr} \left[ \left( \prod_{j=2}^n \frac{1+g_j}{2} \right)^2 \frac{1+s_1}{2} \frac{1 \pm Z_1}{2} \right]  \\
                &= \frac{1}{4} \mathrm{Tr}  \left( \prod_{j=2}^n \frac{1+g_j}{2} \right) = \frac{1}{2} \equiv |a_\pm|^2.
\end{align}%
Thus, we have $\ket{\psi} = a_+ \ket{\psi_+} + a_- \ket{\psi_-}$ with $|a_\pm|=1/\sqrt{2}$. 
But since $\ket{\psi}$ is a stabilizer state, there exists a Clifford unitary $U$ such that 
\begin{align}
U \ket{\psi} = (a_+ \ket{0} + a_- \ket{1}) \otimes \ket{0}^{\otimes n-1},
\end{align}
which is also a stabilizer state. 
The observation from the paragraph above implies the phase difference between $a_+$ and $a_-$ must be a multiple of $\pi /2$. 

Let us consider the action of the $T$ gate on $\ket{\psi_\pm}$ that acts, without loss of generality, on the first qubit.  
The $T$ gate can be decomposed as $T_1 = \alpha \mathds{1} + \beta Z_1$ with $\alpha+\beta = 1$ and $\alpha - \beta = e^{i\pi/4}$. 
Since the generators $g_2, \dots , g_n$ commute with $Z_1$, the action of the $T$ gate depends only on the first generator $\pm Z_1$ yielding $T\ket{\psi_+} = \ket{\psi_+}{}$ and $T\ket{\psi_-} = e^{i\pi/4}\ket{\psi_-}{}$. 
Hence, we find  
\begin{align}
T\ket{\psi} = c_+ \ket{\psi_+} + c_- \ket{\psi_-}
\end{align}
with $c_+ = a_+$ and $c_- = e^{i\pi/4}a_-$, so the phase difference $\arg (c_-/ c_+) = m \pi/2 + \pi/4$ is not a multiple of $\pi/2$. 
Since $\ket{\psi_\pm}$ are stabilizer states with identical generators apart from the sign of the first one, there exists a Clifford unitary $V$ such that 
\begin{align}
    VT\ket{\psi} = (c_+ \ket{0} + c_- \ket{1}) \otimes \ket{0}^{\otimes n-1}.
\end{align}
Assuming $T\ket{\psi}$ is a stabilizer state would imply $VT\ket{\psi}$ is also a stabilizer state. 
However, due to the phase difference between $c_+$ and $c_-$ this cannot be true. 
Thus, we have shown that $T \ket{\psi}$ is a non-stabilizer state, which can be written as a superposition of $\ket{\psi_\pm}$.

\subsection{Monitor form and its retainment \label{app:MonForm}}
Here, we prove the claim that a monitor measurement only stabilizer-purifies a $T$ gate if the corresponding ``logical qubit" is measured. 
We also consider the PBC procedure for this scenario of a single $T$ gate and a monitor $M$ more closely.

\subsubsection{A monitor produces a stabilizer state if and only if it measures the state of the logical qubit}

Here, we show that if a $T$ gate has produced a non-stabilizer state, then a subsequent measurement can only remove the injected magic by measuring this logical qubit. 
That is, it must be a measurement in one of the cosets from Eq.~(\ref{eq:Stab_pur_possibs}). 
It is clear from the result of App.~\ref{app:non-stab_super} that measuring $Z_1$, $s_1$ or $is_1Z_1$ collapses the superposition into a stabilizer state (note $s_1 \ket{\psi_\pm} = \ket{\psi_\mp}$).
Measuring, for example, $Z_1 g$ for $g\in \mathcal{G}$ instead of $Z_1$ does not change anything since $g$ is a stabilizer of both $\ket{\psi_\pm}$.

To show the converse, note that if a measurement $M$ anti-commutes with any $g\in\mathcal{G}$, the post-measurement state is non-stabilizer. 
For example, suppose $\lbrace f, g_2 \rbrace = 0$. 
Then measuring operator $f$ on state $\ket{\psi_\pm}$ and obtaining outcome $\lambda$ results in a state with stabilizer group $\mathcal{S}^\lambda_\pm = \langle \pm Z^\prime_1, \lambda f, h_3, \ldots, h_n\rangle$, where $h_j = g_j$ if $[f,g_j] =0 $ and $h_j = g_j g_2$ otherwise, and similarly $Z^\prime_1 = Z_1$ if $[f,Z_1] = 0$ and $g_2Z_1$ otherwise.
Hence, after measuring $f$ on state $T\ket{\psi}$, we obtain state $c_+ \ket{\psi_+^\lambda} + c_-\ket{\psi_-^\lambda}$ for states $\ket{\psi_\pm^\lambda}$ stabilized by $\mathcal{S}_\pm^\lambda$, respectively. This state is not a stabilizer state.
Therefore, a measurement $M$ that produces a stabilizer state needs to be in the centralizer of $\mathcal{G}$ but it cannot be a member of $\mathcal{G}$ (otherwise the measurement does not change the state); in other words, it is a logical operator with respect to stabilizer group $\mathcal{G}$.

\subsubsection{The monitor measurement is retained in PBC as a single-qubit measurement of the magic state \label{app:Mc}}
We now consider in more detail what happens in the PBC procedure when a single $T$ gate is stabilizer-purified by a monitor $M$, i.e., when the logical qubit introduced by that $T$ gate is measured. 
As before, we replace the $T$ gate with a gadget and introduce an ancillary magic state.
We also introduce dummy measurements of all operators in $\mathcal{S}$ (the stabilizer group of the initial state) that precede all other operations. 
We begin with the case in which the monitor commutes with the gadget measurement $Z_1 Z_a$.
In this case (assuming the monitor measures the logical qubit), it follows from the above that $M = \pm Z_1 g$ for some $g\in\mathcal{G}$. Let us see that this measurement is retained in the PBC circuit as $\pm Z_a$. 

The gadget contains Clifford gates $U = \exp{(-i\frac{\pi}{4}Z_1 X_a)}$ and potentially $S^\dagger_1$. 
Commuting these past the monitor does not change it, since $M = \pm Z_1 g$ commutes with both of these gates.
The GM from the gadget is $Z_1 Z_a$.
We know it anti-commutes with $s_1$; thus it is replaced by some $V = \exp(\lambda \frac{\pi}{4} Z_1 Z_a s_1)$ for $\lambda = \pm 1$ chosen at random. 
Commuting $V$ past $M$ results in measurement operator $\pm s_1 g Z_a$, and restricting this updated operator to the MSR (since it commutes with all the dummy measurements in $\mathcal{S}$) yields a retained monitor $\pm Z_a$.

If the monitor $M$ anti-commutes with the GM $Z_1 Z_a$, we showed that $M=\pm s_1 g$ or $M = \pm i s_1 Z_1 g$. We now show that it is retained in the PBC as either $\pm X_a$ or $\pm Y_a$. 

First, suppose that the gadget measurement outcome is $-1$ so that the $S^\dagger_1$ gate is not included.
Then commuting $U$ past $M$ results in an updated monitor measurement: $\pm i Z_1 s_1 g X_a$ or $\pm s_1 g X_a$, see Eq.~(\ref{eqn:Monitor_mmt_possibs}).
Commuting $V$ (see above) past this measurement results in $\pm g Y_a$ or $\pm s_1 g X_a$, respectively.
These measurements commute with all $\mathcal{S}$ and so may be restricted to the magic state register: they are retained in the PBC as either $\pm X_a$ or $\pm Y_a$.

Second, suppose that gadget measurement outcome is $+1$, so that the gate $S^\dagger_1$ is included, the monitor $M=\pm s_1 g$ or $M = \pm i s_1 Z_1 g$ is updated first to $\pm  is_1 Z_1 g$ or $\pm s_1 g$ respectively, before $U$ and $V$ are commuted past this measurement.
Therefore, the measurement is retained in PBC as either $\pm X_a$ or $\pm Y_a$.

\section{Details of the causal cone of magic \label{app:CCM}}
In this Appendix, we consider the causal set and cone of a $T$ gate in more detail than in Sec.~\ref{sec:CCM}.

\subsection{Definition and SP implications}
We consider one $T$ gate acting on qubit 1 of the $\ket{\psi}$ stabilized by $\mathcal{S} =\langle s_1, \dots , s_n\rangle$. Let $T=T_1$ inject magic due to $\{ Z_1, s_1 \} =0$. The state $T_1 \ket{\psi}$ represents a logical qubit in a magic state encoded by a stabilizer code~\cite{gottesman1997stabilizer}. This code has stabilizer group $\mathcal{G} = \langle g_2, \dots, g_n \rangle$ and logical operators $Z_1$ and $s_1$, where $g_i = s_i$ if $[s_i, Z_1]=0$ or $g_i = s_1 s_i$ if $\{s_i, Z_1\}=0$. 

We are interested in the minimal set of spacetime points outside of which a monitor measurement cannot
SP the time-evolved state for any choice of the Clifford layers after the $T$ gate.
We call this set the \textit{causal set} of a $T$ gate. A measurement $M$ stabilizer-purifies the state $T_1 \ket{\psi}$ if it measures the logical qubit; hence, $M$ must be in one of the cosets (up to a sign) $Z_1 \mathcal{G}, s_1 \mathcal{G}, s_1 Z_1 \mathcal{G}$. cf. App.~\ref{app:MonForm}. 
Thus, the spatial structure of the logical operators from $Z_1 \mathcal{G}, s_1 \mathcal{G}$, and $s_1 Z_1 \mathcal{G}$ determines that of the causal set.

We can sidestep the possibly intricate spatial structure of the causal set, inherited from the logical operators, by defining a causal cone as the convex hull of the causal set.
We call this causal cone the \textit{causal cone of magic} (CCM). 
We first define the apex of the CCM, and then explain how this apex extends causally. 
Note that each member of $s_1\mathcal{G}$ and $s_1 Z_1 \mathcal{G}$ has non-zero support on qubit 1 since $s_1$ anti-commutes with $Z_1$.
We take the most local (i.e., minimal width in periodic boundary conditions) members of $s_1\mathcal{G}$; we then take the union of these most local supports and call that union $A$. (Note that $A$ includes qubit 1.)
Each member of $Z_1 \mathcal{G}$ and $s_1 Z_1 \mathcal{G}$ has non-empty support on $A$ since they anti-commute with each member of $s_1\mathcal{G}$.
Putting all these together, we conclude that any Pauli operator whose support does not intersect $A$ is guaranteed not to be a logical operator.
We define the CCM's apex $\mathcal{I}_A$ as the smallest interval containing $A$. 
Subsequent (non-SP) circuit layers will update the stabilizers and logical operators and this can lead $A$ to expand up to the extent allowed by the Clifford circuit's brickwork structure.
The corresponding causal extension of $\mathcal{I}_A$ defines the CCM. 
The CCM captures the essence of magic spreading while avoiding technical complications arising from the causal set of a $T$ gate.

\subsection{PBC implications}
We further show that a monitor or outcome measurement $P$ can acquire support on the ancillas of two $T$ gates only if it lies in the intersection of their CCMs.

We consider two $T$ gates acting on qubits $1$ and $j$, respectively. We assume that there is no monitor in the layers between them, calling $C$ the combination of Clifford gates between them.
They both add magic only if
\begin{align}
    \{ s_1, G_1 \} = 0 \textrm{ and } \{ g_2, G_2 \}=0,
\end{align}%
where $s_1$ and $g_2$ are distinct stabilizers of $\ket{\psi}$, $G_1=Z_1 Z_{a_1}$ and $G_2 = C^\dagger Z_j Z_{a_2} C$ are the GMs at the time of the $T_1$ gate.
Since both $T$ gates add magic, $Z_1$ and $Z_j$ are not in the stabilizer group at the corresponding $T$ gate's layer, hence neither GM is retained in PBC; both are replaced by a Clifford gate $V_{1,2}$, cf. App.~\ref{app:PBC_Details}.

Going through the PBC procedure, $P$ is updated by commuting past it not only the Cliffords in the circuit but also $U_1, V_1, U_2$ and $V_2$. The measurement acquires support on ancilla $a_i$ only if it anti-commutes with $U_i$, or $V_i$, or both; hence, the measurement must overlap with $U_i$ or $V_i$, and thus have a non-empty support in the CCM of the $T$ gate corresponding to the ancilla $a_i$. 
Therefore, $P$ acquires support on both ancillas only if $P$ is in the intersection of both CCMs.
\section{Bulk monitors SP probability \label{app:bulkMons}}
In this Appendix, we derive the probability of monitors in the bulk of the $T$-circuit-block, i.e., not immediately after the $T$ gate, to stabilizer-purify, in more detail than outlined in Sec.~\ref{sec:SPprob}.

\subsection{Proof of Theorem 2}
Here, we derive Theorem~\ref{thm:gamma}, which we reproduce for convenience.

\setcounter{new_th}{1}
\begin{new_th}
    For a pure stabilizer state, one can choose stabilizer generators such that a single-qubit Pauli operator $M_j$ on qubit $j$ is expressible as
    \begin{align}
        M_j =  \prod_{i=1}^{\gamma_j} g_i^{\alpha_i} \overline{g}_i^{\beta_i},%
        \label{eq:stabDecompApp}
    \end{align}
    up to a $\pm1$ or $\pm i$ prefactor, where $g_i$ are stabilizer generators, $ \overline{g}_i$ are corresponding destabilizers,\footnote{We define a destabilizer $\bar{g}_i$ of generator $g_i$ of stabilizer group $\mathcal{S}$ to be a Pauli operator that anti-commutes with $g_i$ and commutes with all other generators $g_{j\neq i}$ of $\mathcal{S}$~\cite{aaronson2004improved}.} $\alpha_i, \beta_i = 0, 1$, and the number $\gamma_j$ of generators needed satisfies
    \begin{align}
        \gamma_j = 2 S_\mathrm{vN}(j) + \mathcal{O}(1), \label{eq:gammaApp}
    \end{align}
    where $S_\mathrm{vN}(j)$ is the von Neumann entanglement entropy of the subsystem with qubits $1,\dots,j-1$.     
\end{new_th}
\begin{proof}

Using a construction from Ref.~\onlinecite{fattal2004entanglement}, one can always separate the generators of a stabilizer state according to a bipartition with subsystems $B$ and $C$ as (i) local generators in $B$,
(ii) generators straddling the cut, and (iii) local generators in $C$. 
Two other useful results of Ref.~\onlinecite{fattal2004entanglement} are that the minimum number of generators straddling the cut between $B$ and $C$ is twice the entanglement entropy across the cut $2 S_{\mathrm{vN}} (\rho_B) = 2 S_{\mathrm{vN}} (\rho_C) $, and %
\begin{align}
    S_{\mathrm{vN}} (\rho_B) = |B| -  |\mathcal{S}_B|,
\end{align}
where $|B|$ is the size of subsystem $B$ and $|\mathcal{S}_B|$ is the number of generators supported only on $B$.

For the following, define a destabilizer $\bar{g}_i$ of generator $g_i$ of stabilizer group $\mathcal{S}$ to be a Pauli operator that anti-commutes with $g_i$ and commutes with all other generators $g_{j\neq i}$ of $\mathcal{S}$.
Note, for a $(g_i,\overline{g}_i)$ pair, an alternative destabilizer can be defined to be $i g_i\overline{g}_i$. 

We turn to $M_j$ and consider two choices of bipartitions of the qubits. 
First, we put the entanglement cut on one side of $j$ and call subsystem $B$ that which includes qubits $j, \dots, n$.  
Then $M_j$ can overlap only with the $2S_{\mathrm{vN}} (\rho_C)$ generators straddling the cut and with those confined to subsystem $B$ since $j$ is absent from subsystem $B$. 

Second, we put the cut on the other side of $j$ and call subsystem $C'$ the subsystem without $j$, i.e., $C$ is reduced to $C'$ with $|C'| = |C|-1$ qubits.
$\mathcal{S}_{C'}$ differs from $\mathcal{S}_C$ by operators that are either fully supported on qubit $j$ or those that act non-trivially on qubit $j$ and some other qubit(s) in $C'$. 
But we can ensure that there are only at most three such operators, since there are only three non-identity Pauli operators acting on qubit $j$, and any two generators $g_i$ and $g_k$ that act with the same Pauli operator on qubit $j$ can be replaced by $g_ig_k$ which acts as the identity on qubit $j$. 
Therefore $|\mathcal{S}_{C'}| = |\mathcal{S}_{C}| - \mathcal{O}(1)$. 
$M$ does not feature in any element of $\mathcal{S}_{C'}$ so $M_j$ could have featured in only $\mathcal{O}(1)$ of $\mathcal{S}_C$'s generators, i.e., those that are not also generators of $\mathcal{S}_{C'}$. 
This leads to $M_j$ overlapping with at most $\gamma_j$ generators where:
\begin{align}
    \gamma_j = 2 S_{\mathrm{vN}} (\rho_C) + \mathcal{O}(1).
\end{align}%
The Pauli $M_j$ is expressible in terms of these $\gamma_j$ generators and their $\gamma_j$ destabilizers. 
All the other generators (i.e., the $B$ and $C'$ generators) do not have $j$ in their support and thus cannot feature in $M_j$ by themselves.  
Hence, in the expression of $M_j$ in terms of stabilizer generators and destabilizers, the $B$ and $C'$ generators can enter at most as tails tied to the destabilizers featuring in $M_j$. 
This is in order to cancel the destabilizer combination's support in subsystems $B$ and $C'$. 
(Note $B$ and $C'$ destabilizers cannot enter since $M_j$ cannot flip $B$ and $C'$ generators due to not being in their support.) 
However, if these tails including $B$ and $C'$ generators are needed, we can redefine the $BC$ generators such that the tails are removed.
Thus, relabelling the generators and destabilizers entering in $M_j$ for brevity (since these were not necessarily the first $\gamma_j$ generators of the state to be stabilizer-purified) yields
\begin{align}
    M_j =  \prod_{i=1}^{\gamma_j} g_i^{\alpha_i} \overline{g}_i^{\beta_i},%
\end{align}%
up to a $\pm1, \pm i$ prefactor, and $\alpha_i, \beta_i = 0, 1$.
\end{proof}

\subsection{SP probability}
First, we compute the SP probability of a single monitor. 
Consider the two states $\ket{\psi_\pm}$, which differ only in the stabilizer $\pm \tilde{g}_1$.
The favorable cases which lead to SP are 
\begin{align}
    \mathrm{(i)} \quad Z_j &= g_1 \prod_{i=2}^{\gamma_j} g_i^{\alpha_i}, \\ 
    \mathrm{(ii)} \quad Z_j &=  \overline{g}_1 \prod_{i=1}^{\gamma_j} g_i^{\alpha_i},%
\end{align}%
assuming $\tilde{g}_1 \in \{ g_1, \dots, g_{\gamma_j} \}$ and denoting w.l.o.g. $g_1 = \tilde{g}_1$.
In case (i), one of the two states $\ket{\psi_\pm}$ is incompatible with the measurement, while in case (ii), both post-measurement states are the same.
Hence, counting how many combinations of $\bm \alpha$ and $\bm \beta$ lead to SP due to the monitor $Z_j$, excluding the identity, yields
\begin{align}
    \mathbb{P}(Z_j \textrm{  SP}|\,\textrm{prev. NSP}) = \frac{3}{2}\frac{\ 2^{\gamma_j}}{2^{2\gamma_j} -1} \label{eq:ZjSP}
\end{align}%
conditioned on previous monitors not stabilizer-purifying.

Second, we compute the SP probability from all monitors in a single layer, say the $k^\textrm{th}$. 
Different monitors may have different $\gamma_j$s.
Here, we focus on a regime with a volume- or area-law scaling of the EE $S_\mathrm{vN}(j)$.
Thus, each $\gamma_j$ has the same scaling with $n$, that is $\gamma_j = c_j n +\mathcal{O}(1)$ or $\gamma_j = c_j$, with constants $c_j$, for volume- or area-law scaling, respectively.
Since we shall be interested in the scaling of the SP time with $n$, we may simplify the calculations by setting the same  $c_j$ and $\mathcal{O}(1)$ correction for each $\gamma_j$; thus $\gamma_j= \gamma$  for any $j = j_1, \dots, j_w$ where $w$ approximates the total number of qubits in the layer where monitors can potentially SP (see main text). 
The SP probability for the first monitor in the layer conditioned on the previous layers of monitors not stabilizer-purifying is given by Eq.~(\ref{eq:ZjSP}); using $\gamma_j = \gamma$, we denote it as $f \equiv \frac{3}{2}\frac{2^\gamma}{4^\gamma -1}$. 
Similarly, the probability for the second monitor in the layer to stabilizer-purify and the first one not to, conditioned on previous layers not stabilizer-purifying is $(1-f)f$. 
Continuing for all the $pw$ potentially purifying monitors in the $k^\textrm{th}$ layer, the SP probability for this layer, conditioned on previous layers not stabilizer-purifying, is 
\begin{align}
     \mathcal{P} &\equiv \mathbb{P}(d^* = k |\,k'<k \, \textrm{NSP})\\
     &\approx \sum_{i=1}^{pw} f (1-f)^{i-1} = 1 - (1-f)^{pw}. 
\end{align}

Third, we compute the bulk monitors SP probability by considering the SP probability for each layer. 
Similarly to monitors in the same layer, we can treat different layers of monitors as independent apart from the SP conditional. 
Using $\mathbb{P}(d^*=1) = p_1$ from Eq.~(\ref{eq:p1def}), we find the probability for the $k^\textrm{th}$ layer to SP, with $k = 2, \dots, d$%
\begin{align}
    \mathbb{P} (d^* = k) = (1-p_1) (1-\mathcal{P})^{k-2} \mathcal{P}.
\end{align}%
Thus, we find  (for $d\gg1$ in the last step)
\begin{align}
    \mathbb{P}(d\geq d^* >1) &= \sum_{k=2}^d \mathbb{P}(d^* = k) \\
    &= (1-p_1) \left[1-(1-\mathcal{P})^{d-1} \right] \\
    &\approx (1-p_1) \left[ 1 - (1-f)^{pwd}\right].
\end{align}
\section{Spacetime partitioning \label{app:SpacetimePart}}
In this Appendix, we show how the simulation task for each quantum circuit instance can be reduced to simulating a set of smaller circuits by using the structure of monitoring measurements and 2-qubit Clifford gates. 
We denote this procedure \textit{spacetime partitioning}, and map it to an inhomogeneous bond percolation model. 
Using the mapping to percolation, we find a critical monitoring rate $p_\text{c}^\text{TN}$, which marks a phase transition in the simulability by an exact TN contraction. 
A critical monitoring rate $p_\text{c}^\text{PBC}$ of a $\mathrm{\textsf{CPX}_{\mathrm{PBC}}}$ transition is upper bounded by this $p_\text{c}^\text{TN}$, that is $p_\text{c}^\text{PBC} < p_\text{c}^\text{TN}$. %
This upper bound can be interpreted as the analog of the Hartley entropy transition which upper bounds the von Neumann entropy transition~\cite{skinner2019mipt}.
\begin{figure}[t]
    \centering
        \includegraphics[width=8.6cm]{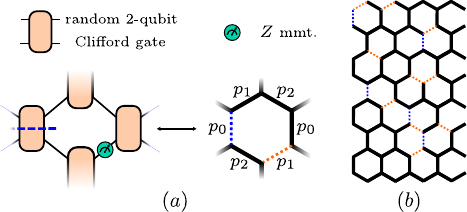}
    \caption{($a$): Equivalent circuit and honeycomb lattice unit cells. Separable gates and monitors, respectively, break vertical and oblique bonds; this leads to bond percolation on the honeycomb lattice with bond probabilities $p_0 = 1-\sigma$ and $p_1 = p_2 = 1-p$. ($b$): Bond percolation lattice corresponding to a circuit as in Fig.~\ref{fig:circuit}$(a)$. The percolation picture is blind to $T$ gates.}
    \label{fig:circ_to_perc}
\end{figure}

\subsection{Spacetime partitioning of the circuit\label{app:stp}}
 Monitors partition the circuit temporally by projecting the state of a qubit and making the previous state partially irrelevant---e.g., in the extreme case of monitoring all the qubits simultaneously at an intermediate time $t_0$, the final state can be reconstructed solely from the monitoring outcomes and the circuit after $t_0$.

Separable Clifford gates partition the circuit spatially.   $C\in\mathcal{C}_2$ is separable if
\begin{equation}
C = u_1 \otimes u_2,    
\end{equation}
where $u_1, u_2 \in \mathcal{C}_1$. 
The classification of the Clifford group reveals~\cite{RBM2012,bravyi2021maslov} that only 576 of the $|\mathcal{C}_2| = 11520$ gates are separable, yielding a separability probability $\sigma = 0.05$. Other $C\in\mathcal{C}_2$ can be decomposed as 
\begin{equation}\label{eq:C2decomp}
    C = \left( u_1 \otimes u_2 \right) \ U \ \left( v_1 \otimes v_2 \right) , 
\end{equation} where $U \in \{ \mathrm{\textsf{SWAP}}, \mathrm{\textsf{CX}}_{1,2}, \mathrm{\textsf{CX}}_{1,2} \mathrm{\textsf{CX}}_{2,1} \}$ and $u_i, v_i$ are in subsets %
of $\mathcal{C}_1$~\cite{RBM2012,bravyi2021maslov}.

Non-separable gates in $\mathcal{C}_2$ coincide with possibly entangling---depending on the input state---gates. Consider a state $\ket{\psi}$ with Schmidt rank $m$ for a bipartition $A$ of the system $\ket{\psi} = \sum_{i=1}^{m}c_i \ket{i_A}\ket{i_B}$. The entanglement (or von Neumann) entropy is bounded by the zeroth R\'enyi (i.e., Hartley) entropy $S_A = S(\mathrm{Tr}_B \ketbra{\psi}{\psi}) \leq S_A^{(0)} = \log m$. Consider the operator-Schmidt decomposition~\cite{nielsen2003quantum} of a 2-qubit gate $U$ acting at the boundary of the bipartition
    $U = \sum_{i=1}^{r} Q_i \otimes R_i$
where $Q_i$, $R_i$ are single-qubit operators and the Schmidt number $r \leq 4$ (by the Hilbert-space dimension of single-qubit operators). 
$U$ can increase the Schmidt rank at most $r$ times: for $\ket{\phi} = U \ket{\psi}$ we have
\begin{equation}
    \ket{\phi} = \sum_{i=1}^m \sum_{j=1}^r c_i (Q_j \ket{i_A})(R_j \ket{i_B}) = \sum_{k=1}^{mr} \mu_k \ket{\Tilde{k}_A} \ket{\Tilde{k}_B}.
\end{equation}
Eq.~\eqref{eq:C2decomp} features $\mathrm{\textsf{CX}_{1,2}} = \ketbra{0}{0}\otimes \mathds{1} + \ketbra{1}{1}\otimes X$ and $\mathrm{\textsf{SWAP}} = \frac{1}{2}\left( \mathds{1} \otimes \mathds{1} + X \otimes X + Y \otimes Y + Z \otimes Z \right)$. 
The Schmidt numbers are 2 for $\mathrm{\textsf{CX}_{1,2}}$ and 4 for $\mathrm{\textsf{SWAP}}$ and $\mathrm{\textsf{CX}}_{1,2} \mathrm{\textsf{CX}}_{2,1}$. 
While $\mathrm{\textsf{CX}}$ is more commonly considered as a possibly entangling gate, let us illustrate that $\mathrm{\textsf{SWAP}}$ can also increase entanglement across a given bipartition. 
Consider two Bell pairs on a bipartite system $\ket{\psi} = \frac{1}{2}\left( \ket{00} + \ket{11} \right)_A \otimes \left( \ket{00} + \ket{11} \right)_B$ which has $S(\rho_A) = \log1$ and a maximum EE across any bipartition $S_{\mathrm{max}} = \log2$. Applying $\mathrm{\textsf{SWAP}}$ on qubits $A_2$ and $B_1$ yields $\ket{\phi} = \mathrm{\textsf{SWAP}} \ket{\psi}$ with
     $\ket{\phi} = \frac{1}{2}\left( \ket{0_{A_1}0_{B_1}} + \ket{1_{A_1}1_{B_1}} \right) \otimes \left( \ket{0_{A_2}0_{B_2}} + \ket{1_{A_2}1_{B_2}} \right)$,
which has the density matrix $\mathrm{Tr}_B \ketbra{\phi}{\phi} = \frac{1}{4} \mathds{1}_4$ with maximum EE across any bipartition $S_{\mathrm{max}}=\log 4$.

Henceforth, we shall regard 2-qubit Clifford gates as either possibly entangling and non-separable ($r=2,4$) or non-entangling and separable ($r=1$).
(Two-qubit gates cannot have $r=3$~\cite{dur2002OpSchmidt}.)
\subsection{Mapping to inhomogeneous bond percolation \label{app:mapping}}
In order to keep track of the spacetime regions of a circuit that fully determine the output state, we note the circuit architecture corresponds to a honeycomb lattice, as depicted in Fig.~\ref{fig:circ_to_perc}($a$): gates and qubit lines correspond to vertical and oblique bonds, respectively. 
The spatial and temporal independence of two regions is modelled by cutting the bonds connecting the respective regions. 
The spacetime partitioning mechanisms described in App.~\ref{app:stp} can induce such independence: monitors cut oblique bonds, while $r=1$ gates cut vertical bonds.   
The probabilistic nature of bond cutting suggests a link to inhomogeneous bond percolation. 
It is inhomogeneous since vertical and oblique bonds are broken with different probabilities, $\sigma \equiv 1 - p_0$ and $p \equiv 1 - p_1 = 1-p_2$, respectively, where $\{p_i\}_{i=0}^2$ are the bond occupancy probabilities.

It is known that the Hartley entropy $S^{(0)}$ of a subsystem can also be mapped to a bond percolation problem; for a brickwork circuit as in Fig.~\ref{fig:circuit}$(a)$, but with generic 2-qubit unitaries instead of Clifford gates, percolation is on the square lattice~\cite{skinner2019mipt}.
The monitored Clifford$+T$ circuits we consider lead to a honeycomb percolation problem instead because, unlike for circuits with Haar-random 2-qubit gates, the $r=1$ gates arise with nonzero probability.\footnote{The square lattice of Ref.~\onlinecite{skinner2019mipt} is recovered for $p_0=1$, i.e., for vanishing probability of $r=1$ gates; then no  vertical bonds are ever cut and contracting these bonds to single points reduces the honeycomb to the square lattice while retaining its connectivity.}
By the relation between entangling properties and $r$ for $\mathcal{C}_2$ gates, computing $S^{(0)}$ for our circuits also maps to the same honeycomb percolation problem, but as we next discuss, honeycomb percolation also leads to $\mathrm{\textsf{CPX}_{\mathrm{TN}}}$. We define $\mathrm{\textsf{CPX}_{\mathrm{TN}}}$ as a runtime proxy for simulating a quantum circuit by an exact tensor network contraction.

\subsection{\texorpdfstring{$\mathrm{\textsf{CPX}}_{\mathbf{\mathrm{\bf{TN}}}}$}{} from percolation}
\label{app:TNtransition}

\subsubsection{Cluster TN contraction runtime \texorpdfstring{$\mathrm{\textsf{CPX}}^{(\mathrm{CC})}_{\mathrm{TN}}$}{}} 

To study $\mathrm{\textsf{CPX}_{\mathrm{TN}}}$, we first consider the clusters connected to the final time boundary of the percolated lattice; we dub these clusters \textit{circuit clusters} (CCs). 
We focus on these clusters as only these enter the simulation of final measurements. 
Our rough runtime estimate for contracting the TN corresponding to a CC with maximal
width $s$ and depth $d$ is 
\begin{equation}
    \mathrm{\textsf{CPX}}^{(\mathrm{CC})}_{\mathrm{TN}} \equiv 2^{\min(s,d)}. 
    \label{eq:TNcpx}
\end{equation}
The $\min(s,d)$ dependence is because the TN for a CC can be contracted either in the temporal or spatial directions, with the runtime scaling exponentially in the number of legs of the TN at each stage of the contraction~\cite{markov2008sim}. This number of legs will be roughly either $s$ or $d$, depending on the direction of contraction. %
The idea of exploiting the shallowest dimension of a quantum circuit in $(2+1)$D was also used in Ref.~\onlinecite{napp2022shallow} to assess the simulability of shallow circuits. $\mathrm{\textsf{CPX}}^{(\mathrm{CC})}_{\mathrm{TN}}$ neglects any $\mathrm{poly}(s,d)$ prefactors and $\mathcal{O}(1)$ prefactors in the exponent; it merely aims for an estimate of whether $\mathrm{\textsf{CPX}_{\mathrm{TN}}}$ may scale exponentially with the system size. In particular, CCs with $\mathrm{min}(s,d) = \mathcal{O} (\log n)$ have $\mathrm{\textsf{CPX}}^{(\mathrm{CC})}_{\mathrm{TN}}= \mathrm{poly}(n)$, so they are efficiently simulable. 
However, $\mathrm{\textsf{CPX}}^{(\mathrm{CC})}_{\mathrm{TN}}$ gives only a sufficiency estimate: e.g., for $q=0$ (i.e., a Clifford circuit) CCs of any $s$ and $d$ are efficiently simulable, yet $\mathrm{\textsf{CPX}}^{(\mathrm{CC})}_{\mathrm{TN}}$ may suggest otherwise.

\subsubsection{Spacetime percolation and \texorpdfstring{$\mathrm{\textsf{CPX}}_{\mathrm{TN}}$}{} 
\label{sec:spacetime_perc} }
In our random quantum circuit problem, $\mathrm{\textsf{CPX}}_{\mathrm{TN}}$ is a measure for typical quantum circuits, hence, it depends on the typical CCs, including their width and depth. 
On top of the standard bond percolation model, we need two additional features: ($i$) a wall-like boundary for the final time, making the lattice semi-infinite in the thermodynamic limit, ($ii$) properties of clusters connected to this boundary, i.e., of CCs. 
The spacetime percolation (STP) model features both of these~\cite{grimmett2018stp}. 

We next use some results from percolation theory to characterize the critical point based on the clusters' properties. The 2D critical surface in the 3D parameter space $\mathbf{p} =(p_0, p_1, p_2) \in [0,1]^3$ of STP is the same as that of standard bond percolation on the same lattice~\cite{grimmett2018stp}. For the honeycomb lattice, this is where the combination
\begin{equation}
    \kappa_{\tinyvarhexagon}(\mathbf{p}) = p_0 + p_1 + p_2 + (1-p_0)(1-p_1)(1-p_2)-2 
\end{equation} is vanishing $\kappa_{\tinyvarhexagon}(\mathbf{p}_{\mathrm{c}}) = 0$~\cite{grimmett2013inhom}.
Since $p_0 = 1 - \sigma = 0.95$ is fixed in our model and $p_1 = p_2 = 1-p$, we find a critical monitoring rate $p_c^\text{TN} \simeq 0.48$. The percolating phase $\kappa_{\tinyvarhexagon}(\mathbf{p}) <0$ and the $\kappa_{\tinyvarhexagon}(\mathbf{p})>0$ phase correspond to the hard and easy to simulate phases (using TN contraction), respectively; henceforth, we refer to these as \textit{easy} ($p>p_{\mathrm{c}}^\text{TN}$) and \textit{hard} ($p<p_{\mathrm{c}}^\text{TN}$).

\emph{Hard phase. }For $p<p_c^\text{TN}$, there exists with overwhelming probability an infinite cluster percolating through 
the lattice~\cite{grimmett1999perc}, as depicted in grey in Fig.~\ref{fig:clusters}($a$), which is the only CC with a significant $\mathrm{\textsf{CPX}}^{(\mathrm{CC})}_{\mathrm{TN}}$. 
It has $\min (s,d) =\mathcal{O}(n)$, thus, implying the hardness (i.e., the exponential scaling with $n$ of $\mathrm{\textsf{CPX}}^{(\mathrm{CC})}_{\mathrm{TN}}$) of simulation by TN contraction. 
CCs with $\min (s,d) = \mathcal{O}(1)$ may also occur but their $\mathrm{\textsf{CPX}}^{(\mathrm{CC})}_{\mathrm{TN}}$ is negligible compared to that of the infinite cluster. 

\emph{Easy phase ($p>p_c^\text{TN}$).} We define the radius of a  CC with width $s$ and depth $d$ as
\begin{equation}
    \mathrm{rad(CC)} \equiv s + d.
\end{equation}
In this phase, also called the subcritical percolation phase, there exists $\lambda > 0$ such that the probability 
for a CC to have a radius larger than $k$ satisfies~\cite{grimmet1991expdecay}
\begin{equation}
    \mathbb{P}_{\mathbf{p}} \left[ \mathrm{rad(CC) }\geq k \right] \leq e^{-\lambda k}, \quad \forall \ k > 0, \label{eq:expDecay}
\end{equation} 
for bond probabilities $\mathbf{p}$. 
Strictly speaking, this  was derived for  hypercubic lattices in $D \geq 2$ dimensions~\cite{grimmet1991expdecay}. However, by the universality of bond percolation in 2D~\cite{grimmett2013universality},  
we expect this result to extend to the honeycomb lattice. 

We first consider the average runtime, 
\begin{align}
    \mathrm{\textsf{CPX}}^{(\mathrm{avg})}_{\mathrm{TN}} &\equiv \mathbb{E}_{\mathbf{p}} \left[ \sum_{\mathrm{CC}} \mathrm{\textsf{CPX}}^{(\mathrm{CC})}_{\mathrm{TN}}\right] \label{eq:avgCPXdef}\\
    &= \sum_{r=1}^{\infty} \mathbb{P}_{\mathbf{p}} \left[ \mathrm{rad}(CC) = r \right] \mathrm{\textsf{CPX}}^{(\mathrm{CC})}_{\mathrm{TN}}, \label{eq:avgCPXdivergent}
\end{align}where in Eq.~(\ref{eq:avgCPXdef}) the sum is over the CCs of a certain realization, while in Eq.~(\ref{eq:avgCPXdivergent}) the sum is over the radius values of any CC. 
Despite the exponential suppression of large CCs, without knowing $\lambda$ (as a function of $p$ and $k$), $\mathrm{\textsf{CPX}}^{(\mathrm{avg})}_{\mathrm{TN}}$ cannot be argued to be poly($n$) due to $\mathrm{\textsf{CPX}}^{(\mathrm{CC})}_{\mathrm{TN}}$ itself exponentially increasing in $r$. 

We can, however, consider the \textit{typical} $\mathrm{\textsf{CPX}}_{\mathrm{TN}}$ instead, defined as
\begin{align}
    \log \mathrm{\textsf{CPX}}^{(\mathrm{typ})}_{\mathrm{TN}} &\equiv \mathbb{E}_{\mathbf{p}} \left[ \sum_{\mathrm{CC}} \log \mathrm{\textsf{CPX}}^{(\mathrm{CC})}_{\mathrm{TN}}\right] \label{eq:typCPXdef} \\
    &= \sum_{r=1}^{\infty} \mathbb{P}_{\mathbf{p}} \left[ \mathrm{rad}(CC) = r \right] \log \mathrm{\textsf{CPX}}^{(\mathrm{CC})}_{\mathrm{TN}}.
\end{align}This can be shown to be finite by successive bounds
\begin{align}
    \log \mathrm{\textsf{CPX}}^{(\mathrm{typ})}_{\mathrm{TN}} &= \sum_{r=1}^{\infty} \mathbb{P}_{\mathbf{p}} \left[ \mathrm{rad}(CC) = r \right] \log\left( 2^{\mathrm{min} (s,d)}\right) \nonumber \\
    &\leq \sum_{r=1}^{\infty} \mathbb{P}_{\mathbf{p}} \left[ \mathrm{rad}(CC) = r \right] \frac{r}{2} \label{eq:sdr} \\
    &\leq \sum_{r=1}^{\infty} \mathbb{P}_{\mathbf{p}} \left[ \mathrm{rad}(CC) \geq r \right] \frac{r}{2} \label{eq:cumul}\\
    &\leq \frac{1}{2} \sum_{r=1}^{\infty} e^{-\lambda r} r = \mathcal{O}(1), \label{eq:lastB}
\end{align}where we applied Eq.~(\ref{eq:expDecay}) in the last line. 
Hence, $\mathrm{\textsf{CPX}}^{(\mathrm{typ})}_{\mathrm{TN}} = \mathcal{O}(1)$: these circuits are easy to simulate.\footnote{A similar bounding scheme yields $\mathrm{\textsf{CPX}}^{(\mathrm{avg})}_{\mathrm{TN}} \leq \sum_{s=1}^{\infty} e^{(2-\lambda) s}$, which is why more knowledge about $\lambda$, and a more accurate runtime proxy, is needed for a conclusion about runtime.}

We focused on CCs due to our focus on final measurements.
Including simulating monitoring measurements would require considering also bulk clusters; the conclusions would be similar: the hard phase we found cannot become easy by having to simulate more measurements and since $\kappa_{\tinyvarhexagon}(\mathbf{p}_{\mathrm{c}}) = 0$ is set by bulk percolation, and since there are at most poly($n$) clusters [taking $D=\text{poly}(n)$], the easy phase would remain efficiently simulable. 

\subsection{\texorpdfstring{$\bm{S^{(0)}}$}{} and \texorpdfstring{$\bm{\mathrm{\textsf{CPX}}}$}{} clusters \label{app:clusters}}

\begin{figure}[t]
    \centering
        \includegraphics[width=8.6cm]{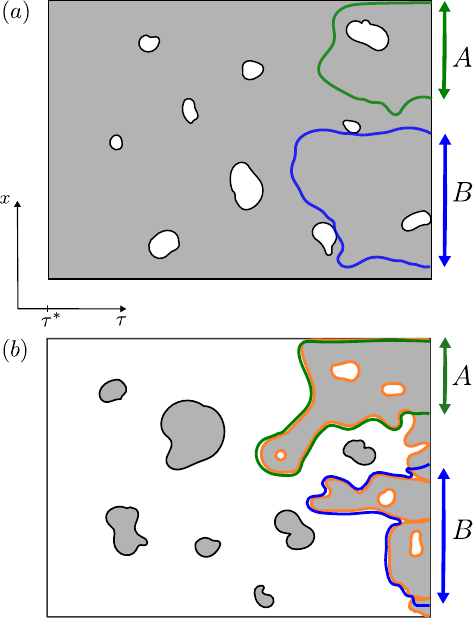}
    \caption{Comparison of $\mathrm{\textsf{CPX}}_{\mathrm{TN}}$ and $S^{(0)}$ clusters. Both ($a$) and ($b$) represent the late-time $\tau > \tau^*$ part of the quantum circuit;  
    the right boundary represents the final state. 
    Grey regions are percolation theory clusters; of these those touching the final state are the $\mathrm{\textsf{CPX}}_{\mathrm{TN}}$ clusters.
    $S^{(0)}$ clusters require choosing subsystems $A, B$ in the final state and are bounded by the corresponding minimal cuts, in green and blue, respectively. 
    ($a$): \textit{hard phase}, $p < p_c^\text{TN}$. The $\mathrm{\textsf{CPX}}_{\mathrm{TN}}$ cluster is the infinite percolating cluster, while each $S^{(0)}$ cluster is finite with minimal cut closed in the bulk. ($b$): \textit{easy phase}, $p>p_c^\text{TN}$. The $\mathrm{\textsf{CPX}}_{\mathrm{TN}}$ clusters (with boundary highlighted in orange) 
    again differ from the $S^{(0)}$ clusters by not being restricted to  subsystems $A,B$.}
    \label{fig:clusters}
\end{figure}

\subsubsection{\texorpdfstring{$S^{(0)}$}{} clusters}
As discussed in App.~\ref{app:mapping}, assuming all 2-qubit gates are not separable implies the bond percolation model is on a
(rotated) square lattice. Paths are defined on the dual lattice. Each step of a path comes with a cost of 0 (or 1) if the crossed bond is empty (or filled). Monitored qubit lines in the quantum circuit, regarded as time-like bonds, correspond to empty bonds. 

The Hartley entropy of a subsystem $A$, denoted by $S_A^{(0)} =  \log \text{rank} (\rho_A )$, corresponds 
to the ``minimal cut'', which consists of the minimal-cost path(s) starting at the boundaries of $A$ at the final time~\cite{nahum2017quantum}. The above cost of 0 or 1 comes from the logarithm of the bond dimension of the qubit line: $0 = \log 1$ if the qubit is monitored or $1 = \log 2$ otherwise. Note the bound $S_A^{(0)} \leq \log \left( \dim \mathcal{H}_A \right) = |A|,$ where $|A|$ is the subsystem size.

When no monitors are present, starting from a pure state, at small enough time $\tau$,
the minimal cut consists of two boundaries spanning the entire time interval $[0,\tau]$. However, at late times there is a discontinuous transition to a single curve closing in the bulk~\cite{nahum2017quantum}. This transition can be understood from the above bound on $S_A^{(0)}$. The minimal cut structure remains the same in the presence of monitors, with a preference for crossing monitored bonds.

Considering the Clifford gates' structure, the bond percolation model is on a honeycomb lattice. Paths are now defined on the triangular dual lattice.
Each step of a path comes with cost 0 (or 1) for a monitored (or not) time-like crossed bond, or with cost $\log r = $ 0, 1, or 2 for a space-like crossed bond replacing a gate with Schmidt number $r$. The minimal cut is defined similarly with the additional feature of moving across space-like bonds. 

The dynamics of interest is in the late-time limit $\tau > \tau^{*} \sim n$~\cite{nahum2017quantum}, where the saturation time $\tau^{*} \sim n$ scaling comes from having a pure input state\footnote{For a maximally mixed input state, the saturation time is $\tau^{*} \sim n^{2/3}$~\cite{yoshida2022ultrafast}.}; hence,
the relevant clusters for $S^{(0)}$ are those bounded minimal cuts closed in the bulk.
Note that there can exist more than one minimal cut for a given subsystem, i.e., multiple cuts with the same minimal cost, which might be nonzero.
\subsubsection{\texorpdfstring{Comparing $\mathrm{\textsf{CPX}}_{\mathrm{TN}}$}{} and \texorpdfstring{$S^{(0)}$}{} clusters}

Circuit clusters selected for an exact TN simulation of the output state have boundaries consisting of strictly 
zero-cost paths\footnote{For open boundary conditions, crossing the system's temporal edge has zero cost [see the top CC in Fig.~\ref{fig:clusters}($b$)].}. These clusters are spacetime regions causally connected to the final time state. 

In contrast to the $S^{(0)}$ clusters, where each choice of an arbitrary subsystem $A$ determines a minimal cut, which defines the $S_A^{(0)}$ cluster, the $\mathrm{\textsf{CPX}}_{\mathrm{TN}}$ clusters are set solely by the circuit and split the system into a set of causally disconnected subsystems. 

Let us consider the distinction between $\mathrm{\textsf{CPX}}_{\mathrm{TN}}$ and $S^{(0)}$ clusters in the late-time behaviour of interest, as depicted in Fig.~\ref{fig:clusters}.
In the hard phase ($p < p_c^\text{TN}$), from percolation theory, we know there exists a unique infinite cluster spanning the entire circuit, which corresponds to the dominant $\mathrm{\textsf{CPX}}_{\mathrm{TN}}$ cluster required for the exact TN simulation of the final state. The minimal cuts for $S^{(0)}$ close in the bulk since we are in the late-time limit. 
Thus, there is a qualitative infinite versus finite distinction between $\mathrm{\textsf{CPX}}_{\mathrm{TN}}$ clusters and $S^{(0)}$ clusters in the hard phase. 
In the easy phase ($p > p_c^\text{TN}$), again from percolation theory, several $\mathrm{\textsf{CPX}}_{\mathrm{TN}}$ clusters are close to the final time boundary. Minimal cuts for $S^{(0)}$ tend to follow the boundaries of $\mathrm{\textsf{CPX}}_{\mathrm{TN}}$ clusters due to their zero cost; however, they are forced to close at the ends of the subsystem for $S^{(0)}$. This constraint prohibits $\mathrm{\textsf{CPX}}_{\mathrm{TN}}$ clusters and $S^{(0)}$ clusters from being identical in the general case despite their significant overlaps in the easy phase.

Although the clusters for $S^{(0)}$ and for $\mathrm{\textsf{CPX}}_{\mathrm{TN}}$ differ, the common underlying percolation model yields the same critical monitoring probability $p_c \simeq  0.48$ for both quantities. The specific effective percolation model of $S^{(0)}$ is the directed polymer in a random environment (DPRE)~\cite{nahum2017quantum}. In contrast, for $\mathrm{\textsf{CPX}}_{\mathrm{TN}}$, the specific model is an extension of STP which considers only clusters with radius scaling faster than $\log(\min (D,n))$, cf. App.~\ref{sec:spacetime_perc}.

\section{Runtime proxy numerical algorithm \label{app:PBC_Numerics}}
\begin{figure}
\begin{algorithm}[H]
  \raggedright{
  Algorithm 1. Computation of MSR blocks and runtime proxy.}
  \algrule
  \algrule
  \textbf{Input}: Circuit size $n$, circuit depth $D$, measurement prob. $p$, $T$ gate prob. $q$, 2-qubit Clifford gates $\mathcal{C}_2$ \\
  \textbf{Output}: MSR block sizes, $\mathrm{\textsf{CPX}}_{\mathrm{PBC}}$%
  \label{EPSA}
   \begin{algorithmic}[1]
   \State Span the circuit realization \textsf{Circuit} and the percolation lattice realization \textsf{PercLattice}
   \State Select \textsf{CC}s, i.e, clusters connected to the final time boundary with minimum dimension $\min (s,d) > \log n$
   \State \textbf{for} each \textsf{CC}
   \Statex \ \ \  \textbf{procedure} Stitching: 
   \State \ \ \ \ \ \ Reconnect the \textsf{CC} as a circuit 
   \Statex \ \ \  \textbf{procedure} PBC:
   \State \ \ \ \ \ \ Replace $T$ gates by $T$ gadgets
   \State \ \ \ \ \ \  Commute Clifford gates to the end of the circuit, 
   \Statex \ \ \ \ \ \  then delete them
   \Statex $\implies$ \textsf{MmtsList} of dummy $Z$s, updated gadget,
   \Statex \ \ \ \ \ \ monitoring, and output measurements
   \State \ \ \ \ \ \  \textbf{for} $M$ in \textsf{MmtsList}
   \State \ \ \ \ \ \ \ \ \ \textbf{if} $\{ M, P \} =0$ for some $P \in $ \textsf{FinalList}:
   \State \ \ \ \ \ \ \ \ \ \ \ \ Replace $M$ with $V(P,M)$
   and commute the 
   \Statex \ \ \ \ \ \ \ \ \ \ \ \ Clifford gate $V$ to the end of the circuit
   \State \ \ \ \ \ \ \ \ \ \textbf{else} 
   \State \ \ \ \ \ \ \ \ \ \ \ \ \textbf{if} $M$ is independent: Append $M$ to \textsf{FinalList}
   \State \ \ \ \ \ \ \ \ \ \ \ \ \textbf{else} Delete $M$
   \Statex $\implies$ \textsf{FinalList}
    \Statex \textbf{procedure} Partitioning: 
    \State \ \ \ By inspecting the supports of the mmts. in \textsf{FinalList},
    \Statex \ \ \ split them into spatially disjoint sets of qubits $\{ \Tilde{R}_a \}$ 
    \State \ \ \ ``Quotient out'' small support mmts. in each $\Tilde{R}_a$  
    \Statex $ \implies$ optimal disjoint sets of qubits $\{ R_i \}$
    \State Compute MSR block sizes, $\mathrm{\textsf{CPX}}_{\mathrm{PBC}}$ 
    \State \textbf{return}  MSR block sizes, $\mathrm{\textsf{CPX}}_{\mathrm{PBC}}$%
   \end{algorithmic}
   \algrule
   \algrule
\end{algorithm}
\end{figure}
In this Appendix, we discuss the algorithm used for computing the MSR block sizes and runtime proxies for the classical simulation of a monitored circuit using PBC.%

We begin by producing a realization of the random circuit 
with a fixed set of Clifford gates from $\mathcal{C}_2$ and locations of monitors and $T$ gates. We also construct the equivalent percolation lattice instance based on the separability of each gate and the monitors, cf. Fig.~\ref{fig:circ_to_perc}($a$). 
Using this percolation lattice, we select the CCs by starting at each vertex on the final time boundary and then we add the neighboring vertices connected to it. 
We retain clusters with both size and depth larger than $\log n$ since they are the only ones that are (potentially) hard to simulate. 
(For the finite sizes tested, this did not produce significant changes in $\textsf{CPX}_\text{PBC}$.) 

The CCs can have boundaries that intersect qubit lines at multiple locations. %
That is, a qubit can be measured at two points within a cluster, with the gates occurring between those measurements not included in the cluster. 
This makes the circuit clusters inadequate for the PBC procedure. We reconnect each circuit by a procedure dubbed stitching. The qubits intersected by a boundary may be monitored with outcome $\lambda_1 = \pm 1$, left idle, then reintroduced in the circuit by a monitor with outcome $\lambda_2$.
If $\lambda_2 = \lambda_1$, then the qubit line between the monitors is directly connected; however, if $\lambda_2 = - \lambda_1$, then a Clifford $X$ gate is inserted between the monitors to flip the qubit, since both monitors are $Z$ measurements.

The stitched circuits corresponding to CCs are amenable to the PBC procedure. 
As discussed in Sec.~\ref{sec:CpxProxy} and App.~\ref{app:PBC_Details}, this is started by commuting all Clifford gates to the end of the circuit, which leads to a list of updated measurements \textsf{MmtsList}. 
Next, we describe how these updates can be implemented for random Cliffords. 

A measurement is a Pauli operator on $n$ qubits
\begin{equation}
    P = p X_1^{x_1} Z_1^{z_1} \dots X_n^{x_n} Z_n^{z_n},
\end{equation} which can be represented as a binary vector of length $2n$ of the form $[x_1 \dots x_n | z_1 \dots z_n]$ and a Boolean phase $p$ (since $P$ must be Hermitian). A Clifford gate $C$ on $n$ qubits is fully determined by a set of $n$ stabilizers $s_1, \dots, s_n$ and a set of $n$ destabilizers $d_1, \dots, d_n$, where each (de)stabilizer is a Pauli operator \cite{aaronson2004improved, Yoder2012AGO, Yoganathan2019QuantumAO} satisfying
\begin{align}
    \left[ s_i, s_j \right] &= \left[ d_i, d_j \right] =  0 \text{ for any $i,j$, and }  \\ 
    \left[ s_i, d_j \right] &= \left\{ s_i, d_i \right\} = 0 \text{ for $i \neq j$}. 
\end{align} Hence, a Clifford can be represented as a $2n \times 2n$ Boolean matrix called its \textit{stabilizer table} (and $n$ Boolean phases) with each row being the vector of the corresponding (de)stabilizer. The (de)stabilizers of $C$ can be thought of in terms of the action $C$ has on single-qubit operators, i.e.,  $d_i \equiv X_i^C = C X_i C^{\dagger}$ and $s_i \equiv Z_i^C =  C Z_i C^{\dagger}$. The stabilizer table of random Clifford $C$ can be efficiently generated; thus, by searching for all $X_i$, $Z_i$ which are present in $P$, one finds the updated measurement after commuting $C$ past $P$
\begin{equation}
    P^C = C^{\dagger} P C = p\  \Tilde{d}_1^{x_1} \Tilde{s}_1^{z_1} \dots \Tilde{d}_n^{x_n} \Tilde{s}_n^{z_n}
\end{equation}
where $\Tilde{d}_i, \Tilde{s}_i$ are the (de)stabilizers of $C^{\dagger}$.

The update of a Pauli operator $M_j$ by a register entangling gate $U= \exp{(-i\frac{\pi}{4} P_i X_{a_i})}$ is solely dependent on the commutation relations with the GM $P_i$. 
The stabilizer table of the Clifford $U$ is not easily accessible, but one can (group) multiply the vectors corresponding to the Pauli measurements $i P_i$, $M_j$, $X_{a_i}$ if $\{P_i, M_j\} = 0$. 

The stabilizer table representation of two measurements can be further used to check their commutation relations efficiently. We use this to reduce \textsf{MmtsList} to a set of mutually commuting, independent measurements \textsf{FinalList}. Initially, \textsf{FinalList} consists of dummy $Z_i$ measurements with $i=1,\dots,n$.
If a measurement $P_2$ with outcome $\lambda_2$ anti-commutes with a measurement $P_1 \in $ \textsf{FinalList} with outcome $\lambda_1$, then it is replaced by a normalized projector $V = (\lambda_1 P_1 + \lambda_2 P_2)/\sqrt{2}$. The Clifford gate $V$ is commuted to the end of the circuit. The update of a measurement $Q \to \Tilde{Q} = V^{\dagger} Q V $ subsequent to $P_2$ can be turned into (group) multiplications of the Pauli operators depending on the commutation relations of $Q$ with $P_1$ and $P_2$. Conversely, if a measurement commutes with all previous elements of \textsf{FinalList}, then we check whether those measurements fully determine it. This independence check is the bottleneck of the algorithm requiring $\mathcal{O}(t^3)$ time~\cite{peres2022pbc, Ko1991AFA}. Then, independent measurements are appended to \textsf{FinalList}. 

We further separate the measurements from \textsf{FinalList} into sets with disjoint supports $\{ \Tilde{R}_a \}$. One can always ``quotient out'' %
single qubit measurements since they commute with higher-weight support measurements, thus, generating additional sets with one element. Other measurements with small support are not guaranteed to lead to more disjoint sets, but one can still optimize over equivalent \textsf{FinalList}s. Finally, using the optimal choice $\{ R_i \}$, we find the MSR block sizes and $\mathrm{\textsf{CPX}}_{\mathrm{PBC}}$, cf. Eq.~(\ref{eq:PBCcpx}).

\section{Table of symbols}

In Table~\ref{tab:symbols} we collect the majority of symbols used throughout this paper. The Sections and context in which these symbols are used are also highlighted.

\begin{table*}
\caption{\label{tab:symbols}Main symbols used in this work. We indicate the context in which they appear and what they denote.}
\begin{ruledtabular}
\begin{tabular}{ccc}
Context & Notation & Meaning \\
\hline
 & $n$ & number of qubits  \\
 & $D$ & total depth of circuit  \\
Circuit model & $p$ & mid-circuit measurement probability  \\
 & $q$ & $T$ gate probability \\
Sec.~\ref{sec:models} & $t$ & total number of $T$ gates in a circuit  \\
 & $C$ & Clifford gate \\
\\
 & $U (U')$ & (adaptive) Clifford gate inside a $T$ gadget  \\
 & $V$ & normalized projector arising in PBC, which is a Clifford gate  \\
Pauli-based  & $M_i$ or $P_i$ & $i^\mathrm{th}$ measurement operator, i.e., a Hermitian Pauli operator  \\
computation (PBC) & $\mathcal{R}_n (\mathcal{R}_t)$ & register of computational (magic ancilla) qubits  \\
 & $\chi_t$ & stabilizer rank of the state $\ket{A}^{\otimes t}$  \\
Sec.~\ref{sec:CpxProxy}  & $\mathsf{CPX}_\mathrm{PBC}$ & runtime proxy of simulating a circuit by PBC  \\
App.~\ref{app:PBC_Details} & $t_i$ & size of the $i^\mathrm{th}$ partition of the magic state register \\
 & $\log \mathsf{CPX}^\mathrm{(typ)}_\mathrm{PBC} /t$ & simulability order parameter of a typical circuit instance \\
\\
 & $\mathcal{S}, \mathcal{S}_\pm \ (\mathcal{G}$) & stabilizer group of a pure (non-)stabilizer state,  \\
 &   & with generators $s_i \ (g_i)$ and destabilizers $\overline{s}_i \ (\overline{g}_i)$   \\
 & $\mathcal{M}$ & generic magic measure  \\
 & $d$ & depth of $T$ circuit block (TCB)  \\
 & $d^*$ & depth of TCB at which SP occurs  \\
 & $\mathbb{P}(\mathrm{SP}) \ [ \equiv \mathbb{P}(\mathrm{SP}, d) ]$ & probability of a TCB to SP (by depth $d$) \\
 & $p_1 \equiv \mathbb{P}(\mathrm{SP}, d=1)$ & probability of a TCB to SP at depth $d=1$ \\  
Stabilizer- & $\{ M_j,$ & \{ measurement operator on the $j^\mathrm{th}$ qubit,  \\
purification (SP) & $\gamma_j,$ & number of (de)stabilizers potentially entering the expression of $M_j$,   \\
 & $S_\mathrm{vN}(j) \}$ & von Neumann entropy across a cut located (w.l.o.g.) to the left of qubit $j$ \} \\
Sec.~\ref{sec:SPandCartoon} & $f \equiv \mathbb{P}(Z_j \textrm{ SP} |\,\textrm{prev. NSP})$ & probability to SP due to $Z_j$ conditioned on no previous SP event  \\
App.~\ref{app:SPTgate}, App.~\ref{app:bulkMons} & $\mathcal{P} \equiv \mathbb{P}(d^* = k |\,k'<k \, \textrm{NSP})$ & probability to SP at depth $k$ conditioned on no SP at smaller depths  \\
 & $w$ & width of the region containing potentially stabilizer-purifying monitors  \\
 & $\Gamma$ & decay rate of $1 - \mathbb{P}(\mathrm{SP}, d)$ with depth $d$ \\
 & $\tau_\mathrm{SP}$ & stabilizer-purification time; $\tau_\mathrm{SP}=\Gamma^{-1}$  \\
 & $r_\mathrm{exp}$ & expected radius of a spacetime volume occupied by one $T$ gate\\
 & $d_\mathrm{exp}$ & expected depth separation between $T$ gates\\
\\
  & $\mathbb{P}(Z_j)$ or $\mathbb{P}(T_j)$ & probability to apply a $Z$ measurement or a $T$ gate on qubit $j$  \\
$T-$correlated & $p_+ = \mathbb{P}(Z_j \vert T_j)$ & probability to apply a $Z$ measurement conditioned on \\
circuit model & & having applied a $T$ gate on qubit $j$  \\
 & $p_- = \mathbb{P}(Z_j \vert  \mathrm{no } \ T_j)$ & probability to apply a $Z$ measurement conditioned on \\
Sec.~\ref{sec:MonGame} & & \textit{not} having applied a $T$ gate on qubit $j$  \\
 & $\alpha = p_+ - p_-$ & partial knowledge parametrization  \\
\\
 & $\mathbf{p}=(p_0, p_1, p_2)$ & vector of bond probabilities on a honeycomb lattice  \\
Mapping to & $\sigma = 1 - p_0$ & probability of a vertical bond to be broken  \\
spacetime percolation & $\mathrm{rad (CC)}$, $s$, and $d$ & radius, maximal width, and depth of a circuit cluster (CC) \\
 & $\mathsf{CPX}^{(\mathrm{CC})}_\mathrm{TN}$ & runtime proxy for an exact the tensor network (TN) contraction of a CC \\
App.~\ref{app:SpacetimePart} & $\mathsf{CPX}^\mathrm{(avg)}_\mathrm{TN}$ & mean TN runtime proxy of a circuit at fixed $\mathbf{p}$\\
 & $\mathsf{CPX}^\mathrm{(typ)}_\mathrm{TN}$ & typical TN runtime proxy of a circuit at fixed $\mathbf{p}$ \\

\end{tabular}
\end{ruledtabular}
\end{table*}

\section{Additional \texorpdfstring{$\bm{qD=\mathcal{O}(1)}$}{} numerics \label{app:add_qD}}

\begin{figure}[t]
    \centering
        \includegraphics[width=8.6cm]{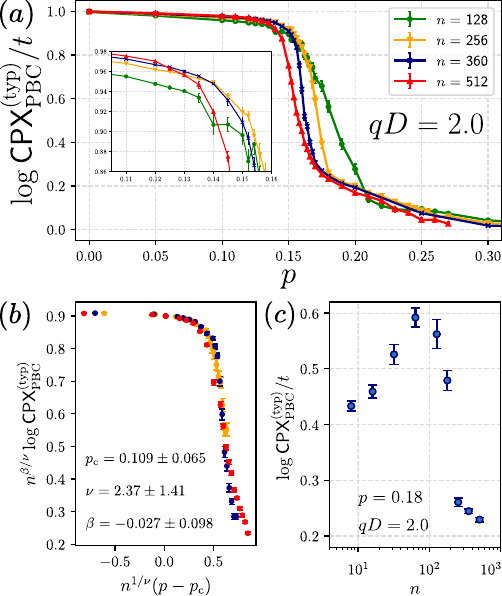}
    \caption{
            Dynamical magic phase transition for $qD=2.0$ consistent with a simultaneous  entanglement  transition. 
            ($a$): The order parameter $\log \mathrm{\textsf{CPX}}^{(\mathrm{typ})}_{\mathrm{PBC}}/t$ (with SE as error bars) versus measurement probability $p$. As $n$ increases, the order parameter approaches a nonzero value for $p < p_\mathrm{c}$ while it approaches zero for $p > p_\mathrm{c}$.
            The inset shows the window in which different $n$ curves cross.
            Significant finite-size effects persist up to at least $n=256$, precluding the accurate extraction of a critical monitoring rate and critical exponents.
            ($b$): Single-parameter
            finite-size scaling collapse, yielding $p_\mathrm{c} = 0.109 \pm 0.065$. 
            ($c$): The order parameter versus system size $n$ at fixed $p=0.18$.
            For intermediate sizes (the data show $n \in \{ 8, 16, 32, 64 \}$) the order parameter seems to increase with $n$; it requires $n\geq 128$  (the data
            show $n \in \{ 128, 180, 256, 360, 512\}$) for the large-$n$ decrease expected for the area-law regime to set in.
            }
    \label{fig:qDapp}
\end{figure}

Here we provide further numerical results for fixed $qD=\mathcal{O}(1)$ in the uncorrelated monitoring model.
We focus on circuits with $qD=2.0$ and depth $D=n$ with $n\leq 512$ qubits. 
Compared with the $qD=0.1$ data, cf. Fig.~\ref{fig:qD_CPX}, we observe strong finite-size effects, cf. Fig.~\ref{fig:qDapp}.
(See the end of this Appendix for an interpretation.)

An indication of the strength of finite-size effects is the drift in curve crossings in Fig.~\ref{fig:qDapp}($a$). 
This makes the accurate extraction of the critical monitoring rate $p_\mathrm{c}$ and critical exponents particularly challenging.

As we did for $qD=0.1$, we start with the scaling ansatz
\begin{align}
    \left[ \log \mathrm{\textsf{CPX}}^{(\mathrm{typ})}_{\mathrm{PBC}}/t \right] (p, n) = n^{-\beta / \nu} G\left( (p-p_\mathrm{c})n^{1/\nu} \right),
\end{align}
where $G$ is a universal scaling function.
Performing a finite-size scaling collapse~\cite{cardyFSS,andreas_sorge_2015_35293,PhysRevLett.28.1516} using this is
shown in Fig.~\ref{fig:qDapp}($b$); the fitted parameters agree within error bars with the $qD=0.1$ results (see Sec.~\ref{sec:CPXtrans}), however these error bars are now considerably larger. 

A possible origin of the drift in crossings, and hence these large fitting error bars, are irrelevant scaling variables. 
For simplicity, we consider incorporating the leading irrelevant variable, leading to the ansatz~\cite{beach2005data,SandvikCorrections,slevinPRL} 
\begin{align}\label{eq:2param}
    \left[ \log \mathrm{\textsf{CPX}}^{(\mathrm{typ})}_{\mathrm{PBC}}/t \right] (p, n) = n^{-\beta / \nu} F( \rho n^{1/ \nu}, u n^{y}),
\end{align}%
where $\rho = p-p_\mathrm{c}$ is the relevant variable (hence $1/\nu>0$) and $u$ is the irrelevant variable (hence $y <0$).
For small values of arguments in $F$, one can recast Eq.~\eqref{eq:2param} as an ansatz with shift and renormalization corrections~\cite{beach2005data}
\begin{align}
    \left[ \log \mathrm{\textsf{CPX}}^{(\mathrm{typ})}_{\mathrm{PBC}}/t \right] (p, n) = &\ n^{-\beta / \nu} (1+cn^{-\omega}) \\
    &\ f\left( (p-p_\mathrm{c})n^{1/\nu} - b n^y \right), \nonumber
\end{align}%
where $b, c$ and $\omega$ are not universal.
Using this ansatz, we find $y = -0.46 \pm 1.96$ consistent with $u$ being irrelevant (the least uncertain $y$ was obtained with fixed $c=0$); however, the accuracy of the $p_\mathrm{c}, \beta, \nu$ estimates does not improve.
Hence, while these corrections should be included, they are still insufficient for an accurate extraction of $p_\mathrm{c}$. 
This limitation, present despite the considerable system sizes accessible to our simulations, is indicative of the strength of finite-size effects for $qD=2.0$.

To further illustrate the strength of finite-size effects, in Fig.~\ref{fig:qDapp}($c$) we plot the $n$-dependence of $\log \mathrm{\textsf{CPX}}^{(\mathrm{typ})}_{\mathrm{PBC}}/t$, focusing on $p=0.18$, i.e., we work in the area-law phase slightly above the entanglement transition. 
$\log \mathrm{\textsf{CPX}}^{(\mathrm{typ})}_{\mathrm{PBC}}/t$ seems to increase for a considerable range of $n$ (up to $n=64$ among the data points shown) before the large-$n$ decrease expected for the area-law phase sets in (as it does for $n\geq 128$ for the data shown). 

The strong finite-size effects can be informally explained by noting that the entanglement-based interpretation of the expected behavior of $\log \mathrm{\textsf{CPX}}^{(\mathrm{typ})}_{\mathrm{PBC}}/t$ rests upon the applicability of the simplified model from Sec.~\ref{sec:SPTime_purifier}. 
This suggests that suppressing finite-size effects requires $d/ \tau_\textrm{SP}$ in Sec.~\ref{sec:SP_AC} to surpass a certain threshold. 
Focusing on the area-law phase, note that for $qD=\eta$ and $D=\mathcal{O}(n)$, we have $r_\text{exp} / \tau_\textrm{SP} = \mathcal{O}(\sqrt{n/\eta})$ for the typical constant-width apex of the CCM. 
[For the atypical $\mathcal{O} (n)$ width apex, we have $d_\text{exp} / \tau_\textrm{SP} = \mathcal{O}(n/\eta)$ if $w_0 = \mathcal{O}(n)$.]
Thus, we find that an $M$-fold increase of $\eta$ implies an $M$-fold increase in the value of $n$ required to suppress finite-size effects. 
In particular, for $qD=2.0$, we would need $20$ times larger systems to study the transition than for $qD=0.1$ to get comparable data to those in Fig.~\ref{fig:qD_CPX}. 
\providecommand{\noopsort}[1]{}\providecommand{\singleletter}[1]{#1}%

\end{document}